\documentclass[aps,prx,twocolumn,superscriptaddress,groupedaddress] {revtex4-2}

\usepackage{amsfonts} 
\usepackage{dsfont}
\usepackage[dvipsnames]{xcolor}
\usepackage[colorlinks]{hyperref}
\usepackage{amsmath}
\usepackage{amsthm}
\usepackage{xfrac}
\usepackage{physics}
\usepackage[capitalise]{cleveref}
\usepackage{graphicx}
\usepackage{mathtools}
\usepackage[caption=false]{subfig}
\usepackage{tikz}
\usepackage{tikz-3dplot}

\usepackage{comment}

\usepackage{orcidlink}

\definecolor{UniBoRed}{HTML}{CE4220}
\definecolor{EthBlue}{HTML}{215caf}
\hypersetup{colorlinks,
    linkcolor=UniBoRed,
    filecolor=UniBoRed,
    urlcolor=UniBoRed,
    citecolor=UniBoRed
}
\newtheorem{theorem}{Theorem}
\newtheorem{lemma}{Lemma}
\newtheorem{proposition}{Proposition}

\newcommand{\mps}{\mathrm{MPS}}
\newcommand{\camps}{\mathrm{CAMPS}}
\newcommand{\sx}{X}
\newcommand{\sy}{Y}
\newcommand{\sz}{Z}
\newcommand{\vn}{\hat{\boldsymbol{n}}}
\newcommand{\vsigma}{\hat{\boldsymbol{\sigma}}}
\newcommand{\so}[1]{\text{SO}(#1)}

\newcommand{\phantomsubcaptionlabel}[1]{\refstepcounter{subfigure}\label{#1}}
\DeclareMathOperator{\conv}{\text{conv}}
\DeclareMathOperator{\sym}{\text{sym}}

\begin{document}

\author{Lorenzo Fioroni\,\orcidlink{0009-0006-6824-2665}}
\author{Filippo Ferrari\,\orcidlink{0009-0003-6317-0816}}
\email{filippo.ferrari@epfl.ch}
\author{Emanuele Tirrito\,\orcidlink{0000-0001-7067-1203}}
\affiliation{Institute of Physics and Center for Quantum Science and Engineering,\\ \'Ecole Polytechnique F\'ed\'erale de Lausanne (EPFL), Lausanne, Switzerland}

\title{Noise-induced classical phases in optimally-unraveled random quantum circuits}

\date{\today}

\begin{abstract}
    We study the classical simulability of open quantum dynamics using random Clifford circuits doped with non-Clifford phase rotations and subject to local noise. 
    We unravel the dynamics into stochastic quantum trajectories simulated with Clifford-augmented matrix product states, and introduce a simulation cost that quantifies the classical resources required. 
    Optimizing this cost over stochastic unravelings, we identify noise-induced classical phases: extended parameter regions in which the dynamics can be fully disentangled by Clifford operations at arbitrary circuit depth.
    Their existence depends on both the noise model \emph{and} the unraveling. 
    Using a geometric representation of quantum channels, we analytically determine optimal unravelings for a broad class of noise models, with numerical simulations confirming the predicted phase boundaries. 
    We further relate the optimal cost to the unraveling-independent nonstabilizerness of the channel and show that, together with trajectory-resolved entanglement and nonstabilizerness, it classifies distinct dynamical regimes. 
    Finally, we show that these classical phases disappear in the averaged density-matrix description, where no unraveling freedom remains. 
    Our results show that the emergence of classicality in noisy random circuits depends on the measurement scheme adopted to probe it, and paves the way to further studies on the classical simulability of driven-dissipative dynamics.
\end{abstract}

\maketitle

\section{Introduction}
Locating the boundary between quantum evolutions that can be efficiently reproduced on a classical computer, and those genuinely requiring quantum hardware is a central problem in quantum information science.
On the one hand, quantum computers promise to outperform classical machines for various tasks~\cite{arute2019quantum, morvan2024phase, liu2025certified}, ranging from cryptography~\cite{simon1997power, shor1997polynomial} to the simulation of quantum many-body dynamics~\cite{feynman1982simulating, lloyd1996universal, daley2022practical}.
On the other hand, classical simulation has played an equally important role in sharpening such boundary, since several classes of quantum evolutions do admit efficient classical descriptions.
The most prominent examples are stabilizer circuits, highly entangled yet efficiently simulable via the Gottesman-Knill theorem~\cite{gottesman1998heisenberg, aaronson2004improved},
and weakly entangled states admitting a compact representation within the tensor networks formalism~\cite{verstraete2008matrix, orus2014practical, cirac2021matrix}, such as matrix product states (MPS)~\cite{white1992density, white1993density, vidal2003efficient, eisert2010colloquium, schollwock2011density}.
These tractable wave functions occupy isolated corners of the Hilbert space, and each corner comes with its own currency to depart from it.
For stabilizer-based simulations, the canonical quantity is nonstabilizerness~\cite{bravyi2016trading, bravyi2019simulation}, with analogous constructions for mixed states~\cite{howard2017application} and quantum channels~\cite{seddon2019quantifying, seddon2021quantifying}.
Tensor-network complexity is instead controlled by the bond dimension~\cite{schollwock2011density}, linked to entanglement.
Classical simulation complexity is therefore governed not by the dimension of the Hilbert space, but by the specific resource that obstructs a chosen efficient representation.
So that, the resource controlling simulation complexity depends on the structure of the classical representation.

The role of these computational resources becomes particularly relevant in realistic quantum devices, where noise can modify the structures that determine classical simulability. Any quantum processor is subject to noise which corrupts quantum information processing~\cite{shor1995scheme, steane1996error, nielsen2012quantum}, while error correction~\cite{terhal2015quantum} and mitigation~\cite{cai2023quantum} schemes can help recover fragile quantum states.
More generally, noise results in an effective dynamics which can deviate qualitatively from the paradigm of unitary quantum mechanics~\cite{zurek2003decoherence, breuer2007theory, rivas2012open}.
A crucial question is whether noise can also make quantum dynamics \emph{easier} to reproduce classically.
This has motivated extensive work on measurement- and noise-induced transitions in simulation complexity ~\cite{li2018quantum, skinner2019measurement, bao2020theory, choi2020quantum, dowling2026noise, shao2026complexity}, as well as classical algorithms that exploit noise to suppress correlations, operator complexity, or the effective circuit depth~\cite{cheng2021simulating,noh2020efficient, fontana2025classical, schuster2025polynomial, martinez2025efficient, angrisani2026simulating, mele2026noise}.

\begin{figure}[t!]
	\centering
	\includegraphics[page=2]{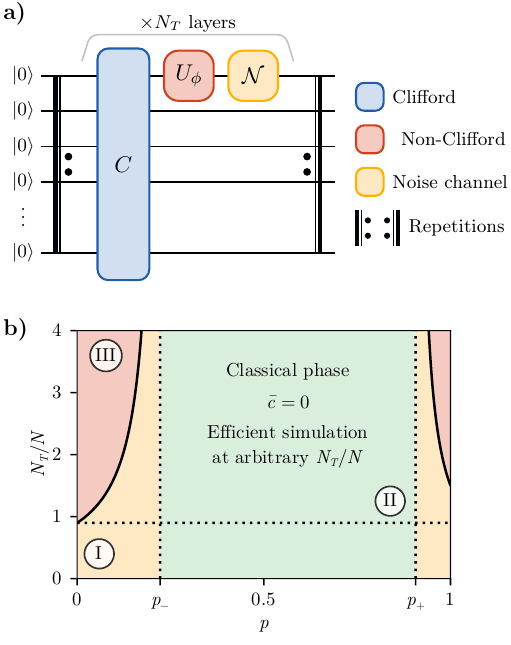}
	\caption{
		Noisy doped random Clifford circuits and their phase diagram.
		(a) One of the $N_T$ layers: a global random Clifford $C$ acting on the $N$ qubits, followed by the non-Clifford rotation $U_\phi$ and by the local noise channel $\mathcal N$, both applied to the first qubit.
		(b) Sketch of the phase diagram in the plane of the noise strength $p$ and of the non-Clifford density $N_T/N$.
		Between the threshold probabilities $p_\pm$ (vertical dotted lines) the optimal unraveling has $\bar c = 0$: the trajectories stay Clifford-disentanglable at arbitrary depth and the circuit is in the noise-induced classical phase (II).
		Outside this window, disentangling is possible only below the boundary $N_T \sim N/\bar c$ (solid line), which separates the disentanglable regime (I) from the quantum one (III).
		The three regions are characterized throughout the paper in terms of the entanglement and nonstabilizerness of the trajectories.
	}
	\label{fig:sketch}
	{\phantomsubcaptionlabel{fig:sketch_circuit}}
	{\phantomsubcaptionlabel{fig:sketch_phases}}
\end{figure}

Noisy quantum dynamics admits two complementary descriptions.
In the unconditional picture, the system is represented by a mixed density operator evolving under a quantum channel or, in continuous time, a Lindblad master equation~\cite{lindblad1976generators, breuer2007theory}.
Alternatively, the same dynamics can be \emph{unraveled} into an ensemble of stochastic pure-state trajectories~\cite{dalibard1992wave,plenio1998quantum,daley2014quantum} whose average reproduces the density operator.
Crucially, the unraveling is not unique: different Kraus decompositions of the same channel generate distinct trajectory ensembles while leaving the averaged dynamics unchanged.
Physically, this freedom is the choice of how the environment is monitored; computationally, it is a gauge that can be set so as to make individual trajectories as cheap as possible.
So far, this gauge has been fixed by suppressing trajectory entanglement, either by optimizing it directly or through related local criteria~\cite{vovk2022entanglement,cheng2023efficient,kolodrubetz2023optimality,chen2024optimized,vovk2024quantum, daraban2025nonb, cichy2026classical}.
This is the natural choice for tensor-network simulations, where the computational cost of a trajectory is controlled by the entanglement it carries, and it substantially enlarges the regime accessible to MPS-based simulations of noisy random circuits and dissipative dynamics~\cite{cheng2023efficient,chen2024optimized,cichy2026classical}.
Entanglement, however, is a proxy for cost only within one family of representations.
In a simulator that already exploits efficiently tractable Clifford structures, highly entangled correlations may be computationally inexpensive, and the obstruction lies elsewhere: in the non-Clifford operations that the stabilizer formalism cannot absorb.

In this work we provide a general understanding of the classical simulability of open quantum systems by redefining the entanglement-based optimal unraveling paradigm.
In practice, we answer the following three questions:
\begin{itemize}
	\item[(i)] Can the unraveling freedom of an open quantum system be fixed by targeting the resource that controls a chosen classical representation, rather than entanglement alone?
	\item[(ii)] Can such a resource-adapted unraveling stabilize an extended region in which arbitrarily large and deep universal circuits remain classically simulable?
	\item[(iii)] Once the optimal unraveling is fixed, how is it related to trajectory-resolved quantum resources such as entanglement and nonstabilizerness?
	      Does either of them, on its own, track the simulation cost?
\end{itemize}
We address (i), (ii) and (iii) using random Clifford circuits doped with single-qubit non-Clifford phase rotations and subject to local noise (see \cref{fig:sketch_circuit}).
These circuits interpolate between efficiently simulable stabilizer dynamics and universal quantum computation while retaining enough structure to isolate the origin of their simulation complexity analytically.
We follow their open-system dynamics through stochastic quantum trajectories represented using Clifford-augmented matrix product states (CAMPS), in which a many-body state is written as a Clifford circuit acting on an MPS~\cite{qian2024augmenting}.
The Clifford component stores correlations that can be handled efficiently within the stabilizer formalism, while the inner MPS carries only the residue that no Clifford transformation can remove.

The paper is structured as follows.
In \cref{sec:results} we outline the main results of our work, which amount to a detailed answer of questions (i), (ii) and (iii).
In \cref{sec:setup_methods} we introduce the noisy random quantum circuit and the CAMPS Ansatz.
In \cref{sec:classical_phases} we introduce the simulation cost, construct the system phase diagram for various noise models and derive the analytical results.
In \cref{sec:magic} we connect the simulation cost to known nonstabilizerness measures.
In \cref{sec:campo_disentangling} we inspect whether the disentangling power of quantum trajectories extends also to the averaged density matrix.
Finally, we draw our conclusions and discuss future directions in \cref{sec:discussion}.

\section{Overview of main results}
\label{sec:results}

In this work, we introduce a rigorous notion of ``simulation cost", i.e., the amount of classical resources needed to simulate dissipative dynamics efficiently with CAMPS.
We show that such cost is a property of the noise channel \emph{and} of the way it is unraveled, not of the states the trajectories visit.
To make this concrete, we split the single-qubit channel as a combination of Clifford and non-Clifford operations, and define the cost $\bar c$ as the smallest fraction of non-Clifford operations for which such a splitting exists.
Operationally, $\bar c$ is the probability that a trajectory picks up a non-Clifford operation at each noise event.
$1-\bar c$ represents instead the complementary probability that the dynamics acts as a Clifford gate, tracked for free by the stabilizer component of the CAMPS.
The optimal unraveling therefore extends the depth over which trajectories can be completely Clifford-disentangled by a factor set by the inverse non-Clifford weight.
Most importantly, when $\bar c=0$, the noisy channel admits an unraveling built entirely out of Clifford operations, so that trajectories remain Clifford-disentanglable at arbitrary circuit depths.
This defines a \emph{noise-induced classical phase}.

Using a geometric representation of single-qubit quantum channels, we solve this optimization analytically for a broad class of Pauli noise, including single-qubit dephasing and depolarizing noise.
The geometry makes the answer transparent: it identifies the convex region generated by Clifford operations, and determines both the optimal cost, and the unraveling that realizes it.
CAMPS simulations quantitatively reproduce the analytical phase boundaries.
Crucially, the emergent classicality is a property of the unraveling rather than of the channel alone: the same physical noise, with a naive Kraus decomposition, hides the classical phase entirely.
Moreover, noise does not generically improve simulability.
For tilted dephasing, the interplay between the noise axis and the non-Clifford rotation instead raises the simulation cost and hinders Clifford disentangling.
Classical accessibility is therefore controlled by the structure of the optimally unraveled channel, rather than by noise strength alone.

The resulting quantum dynamics organizes into three regimes according to the optimal cost $\bar c$ and the quantum resources we analyze, namely entanglement and nonstabilizerness, as summarized in \cref{fig:sketch_phases}.
A crucial result shining light on the optimal cost $\bar c$ we introduce is its analytic connection with nonstabilizerness: in general, the cost $\bar c$ of the optimal unraveling is a bound for the unraveling-independent measure of nonstabilizerness called robustness of magic introduced in Ref.~\cite{howard2017application}, and for specific noise models the two quantities are equivalent up to a scaling factor.
When looking at trajectory-dependent resources, namely the Von Neumann entropy and stabilizer Renyi entropy~\cite{leone2022stabilizer}, we identify the distinct features of the three aforementioned dynamical phases.
In the \emph{disentanglable} phase, the inner-MPS entanglement vanishes while the optimal cost $\bar c$ and the stabilizer R\'enyi entropy remain finite: nonstabilizerness is nonzero but does not obstruct complete Clifford disentangling.
In the \emph{classical} phase, $\bar c = 0$ and both the inner-MPS entanglement and its stabilizer R\'enyi entropy vanish.
In the \emph{quantum} phase, $\bar c > 0$ once again, yet residual entanglement and nonstabilizerness now grow with circuit depth, and the resources required by CAMPS grow with them.
Nonstabilizerness alone therefore does not measure the simulation cost any more than entanglement does.
This second series of results provides a complete characterization of the dissipative dynamics of our simple yet universal noisy toy model.

Finally, we demonstrate that the unraveling procedure is \emph{necessary} for the appearance of classical phases in the considered setup.
Indeed, by extending the CAMPS Ansatz to density operators, which model the averaged dynamics of open systems, we show that the disentangling power of the optimally-unraveled CAMPS is lost.
More broadly, our work represents, to the best of our knowledge, the first application of Clifford-augmented tensor networks to open quantum systems, and paves the way for future studies on the complexity and simulability of dissipative quantum dynamics.

\section{Setup and methods}
\label{sec:setup_methods}

\subsection{Noisy doped random Clifford circuits}

Throughout this paper we consider a quantum circuit with $N$ qubits and $N_T$ layers of gates.
Each layer is composed of a random, global Clifford circuit $C$, followed by the application of a single non-Clifford rotation $U_\phi$ about the $z$ axis by an angle $\phi$, and an arbitrary local noise channel $\mathcal N$ on one of the qubits.
Due to the presence of $C$, we can assume without loss of generality that the non-Clifford operations are both applied on the first qubit.
Indeed, a suitable selection of the Clifford circuit can map this case to any other possible choice.
We model the possible errors following the application of the non-Clifford operation via the local noise channel $\mathcal N$.
A sketch of the system is presented in \cref{fig:sketch_circuit}.
The motivation behind the choice of this setup is that it interpolates between stabilizer evolution -- efficiently simulable on a classical computer~\cite{gottesman1998heisenberg, aaronson2004improved} -- and universal quantum computation~\cite{nielsen2012quantum}.

We consider local noise channels $\mathcal N$ in the form
\begin{equation}
	\mathcal N(\rho) = \sum_j p_j K_j \rho K_j^\dagger,
\end{equation}
where $\sum_j p_j = 1$ and $ K_j$ are local unitary Kraus operators.
Equivalently, we define the quantum channel that first applies the non-Clifford rotation $ U_\phi = e^{i \phi \sz}$ and then the local noise as
\begin{equation}
	\label{eq:lambda_channel}
	\Lambda (\rho) = \sum_j  p_j K_j e^{i \phi \sz} \rho e^{-i \phi \sz}  K_j^\dagger = \sum_j p_j V_j \rho V_j^\dagger.
\end{equation}
Here and in the rest of the paper we refer to the Pauli operators as $\sx$, $\sy$ and $\sz$, respectively.
The discrete time evolution described by \cref{eq:lambda_channel} refers to the unconditioned density matrix $\rho$.
Equivalently, the same map can be recast in a stochastic evolution for the pure state $\ket{\Psi}$ by unraveling the channel in quantum trajectories~\cite{breuer2007theory, daley2014quantum}.
At each time step, one of the Kraus operators $V_j$ in the channel is randomly selected according to its probability $p_j$ and is applied to the state of the $k$-th trajectory $\ket{\Psi_k} \mapsto V_j \ket{\Psi_k}$.
Averaging over the trajectories $\ket{\Psi_k}$ recovers the averaged dynamics generated by the channel.

\subsection{Clifford-augmented matrix product states}

Matrix product states (MPS) provide a compact representation of many-body wave functions $\ket{\Psi}$ and became a standard method for the classical simulation of quantum many-body systems~\cite{schollwock2011density}.
The MPS formalism takes advantage of the system's topology, and in particular of the locality of the couplings.
An arbitrary $N$-qubit state $\ket{\Psi} = \sum_{\sigma_1, \ldots, \sigma_N} \Psi_{\sigma_1, \ldots, \sigma_N} \ket{\sigma_1 \ldots \sigma_N}$ is approximated by factorizing the coefficient tensor $\Psi$ into a product of low-rank matrices:
\begin{equation}
	\ket{\Psi_\mps} = \sum_{\sigma_1, \ldots, \sigma_N} A^{[\sigma_1]} \ldots A^{[\sigma_N]}\ket{\sigma_1\ldots\sigma_N}.
\end{equation}
The rank $\chi$ of each tensor $A^{[\sigma_j]}$ is called \emph{bond dimension}, and controls both the scalability and the accuracy of the Ansatz.
Low-entangled states can be represented efficiently with an MPS with small bond dimension, whereas states characterized by an entanglement that grows with the system size can require an exponentially large $\chi$.
Therefore, entanglement directly determines whether an MPS will or won't be a suitable Ansatz to describe a given system.

However, not all entangled states are created equal.
It is useful to differentiate between the kind of entanglement that can be treated efficiently on a classical computer -- i.e., Clifford -- from the one that is out of reach to classical algorithms as the system size increases.
A formalism that harnesses this distinction explicitly is that of Clifford-augmented MPS (CAMPS)~\cite{qian2024augmenting, masotllima2024stabilizer, nakhl2025stabilizer}.
The CAMPS Ansatz represents a state $\ket{\Psi}$ as the application of a Clifford operator $C$ onto an MPS:
\begin{equation}
	\ket{\Psi_\camps} = C \ket{\Psi_\mps}.
\end{equation}
The Clifford operator describes entanglement within the stabilizer formalism, while the inner MPS handles the remaining quantum correlations.
As such, a state that is highly entangled by Clifford circuits only, will have a bond dimension of the inner MPS $\chi = 1$.
Stated differently, if the state resulting from the application of arbitrary circuit is represented by a CAMPS with bond dimension $\chi=1$, the circuit is efficiently simulable on a classical computer.
Notably, this property does not hold for the MPS Ansatz: even if the circuit induces only stabilizer entanglement, the bond dimension can grow exponentially in the system size.
We refer to the ability of ``moving" entanglement from the MPS to the outer Clifford operator as \emph{disentangling} the MPS.
Thus, CAMPS are capturing an increasing attention for the study of quantum many-body systems~\cite{qian2025clifford, mello2025clifford, fan2025disentangling,  qian2025augmenting, huang2025clifford, huang2025nonstabilizerness, li2025disentangling, bacciconi2026pulling}.

\subsection{Closed-system limit}
\label{sec:closed_system}

The circuit in \cref{fig:sketch_circuit} has been studied in the closed-system (i.e., noiseless) limit in Refs.~\cite{fux2025disentangling, liu2026classical,MasotLima2026,bejan2024}, with $T$-gates (corresponding to $\phi=-\pi/8$) as non-Clifford operations.
In that case, it has been shown that a CAMPS Ansatz can disentangle the inner MPS up to a number of $T$ gates $N_T \sim N$.
In this work, we ask whether noise sources can lead to the appearance of \emph{classical} phases, namely wide parameter regions where CAMPS can disentangle the MPS independently of the circuit depth.

\section{Classical phases in doped random Clifford circuits}
\label{sec:classical_phases}

Given the composite channel in \cref{eq:lambda_channel}, we seek an optimal unraveling that minimizes the overhead from non-Clifford operations.
We consider a convex decomposition of the form
\begin{equation}
	\Lambda(\rho) = \bar c \, \Phi(\rho) + (1 - \bar c) \, \Sigma(\rho),
	\label{eq:cost}
\end{equation}
where $\Sigma$ is a Clifford channel and $\Phi$ denotes the non-Clifford component.
Since Clifford operations can be efficiently tracked within the CAMPS framework, minimizing the weight $\bar c$ directly minimizes the classical simulation cost.
Specifically, as mentioned in \cref{sec:closed_system}, CAMPS can disentangle up to a number of non-Clifford operations $N_T \sim N$.
If an unraveling with cost $\bar c$ is found, the average number of non-Clifford operations applied up to time $N_T$ will be effectively rescaled as $\bar c \, N_T$.
Consequently, CAMPS will be able to disentangle the state up until $N_T \sim N/\bar c$.
Notably, the limiting case $\bar c = 0$ corresponds to an unraveling where the state remains fully disentangled by Clifford operations at arbitrary circuit depths -- the regime we identified as classical phase.
In the remainder of this section, we prove that depending on the noise model such phases can appear, and establish rigorous results characterizing their boundaries and properties.

\subsection{Aligned dephasing}
\label{sec:aligned_dephasing}

First, we consider the circuit in \cref{fig:sketch_circuit} under the action of a dephasing channel
\begin{equation}
	\label{eq:aligned_dephasing}
	\mathcal N_{p}(\rho) = (1-p) \rho + p \sz \rho \sz,
\end{equation}
where $p \in [0,1]$ is the associated probability.
Notice that this noise channel preserves the rotation axis of the non-Clifford operator $U_\phi$, and thus it is referred to as \emph{aligned dephasing}.
A more generic setting that makes no assumption on the relative orientation between the rotation axis of $U_\phi$ and the direction preserved by the channel is studied in \cref{sec:tilted_dephasing}.

\subsubsection{Geometric representation of the channel}
\label{sec:geometric_aligned_dephasing}

We introduce a geometric representation of the channel $\Lambda$ that will prove useful in finding the optimal unraveling according to \cref{eq:cost}.
To do so, we consider the Pauli Transfer Matrix (PTM) of the channel $L_{jk} = \sfrac{1}{2} \Tr(\sigma_j \Lambda(\sigma_k))$, where $\sigma_j \in \{\mathds{1}, X, Y, Z\}$.
Computing $L$ for the channel $\Lambda$ associated to aligned dephasing noise, we find the block-diagonal form
\begin{equation}
	\label{eq:aligned_dephasing_ptm}
	L_\Lambda = \begin{pmatrix}
		1 &  & \\ & f R(2\phi) &\\ &&1
	\end{pmatrix}; \quad
	R(2\phi) = \begin{pmatrix}
		\cos(2\phi) & \sin(2\phi) \\ -\sin(2\phi) & \cos(2\phi)
	\end{pmatrix},
\end{equation}
where $f = 1-2p$, and $R \in \so{2}$.
Notice that the PTM preserves the components relative to $\mathds{1}$ and $\sz$, while rotates and scales by a factor $f$ those relative to $\sx$ and $\sy$.
The scaling is due entirely to the noise channel, as $\mathcal N_p(\sx) = f \sx$ and $\mathcal N_p(\sy) = f \sy$.
We encode the $XY$ block of $L$ in the complex number $z = f \, (\cos(2\phi) + i \sin (2\phi))$, which we represent on the complex plane.
Notice that only when $p=0$ or $p=1$, $|z| = 1$ and thus the associated point lies on the unit circle.
Otherwise, $|z| < 1$ and the associated point lies \emph{inside} the unit circle, on the line of angle $2\phi$ with respect to the real axis.

As we prove in \cref{sec:proofs_aligned_dephasing}, the optimization problem for the minimal $\bar c$ is itself defined on the complex plane, as each one of the Kraus operators is a rotation about $z$, and can be described by a single complex number.
Among the $z$ rotations that can be represented in the complex plane, we identify the operators $\mathcal G = \{S, \sz, S^\dagger, \mathds{1}\}$, that also belong to the Clifford group.
These form a cyclic group generated by the phase gate $S$.
By writing the $XY$ blocks of the PTMs of the operators in $\mathcal G$, we can derive their coordinates in the complex plane:
\begin{equation}
	z_S = -i,\quad z_Z = -1 ,\quad z_{S^\dagger} = i,\quad z_{\mathds{1}} = 1.
\end{equation}
That is, the four Clifford operators lie on the vertices of the diamond inscribed in the unit circle.
More generally, we allow the channel $\Sigma$ in \cref{eq:cost} to be a convex combination of Clifford channels $\Sigma_j$:
\begin{equation}
	\begin{split}
		\Sigma(\rho) = & \sum_j \alpha_j \Sigma_j(\rho), \quad \sum_j \alpha_j=1, \\& \quad \Sigma_j(\rho) = G \rho G^\dagger,
	\end{split}
\end{equation}
with $G \in \mathcal G$.
This corresponds to taking convex combinations of the associated coordinates in the complex plane, and thus spans the whole diamond including its interior.

Any quantum channel that can be represented as a point on the Clifford diamond, allows for an unraveling made of Clifford operators only, and thus its cost is $\bar c = 0$.
Conversely, a noisy circuit whose coordinates are inside the unit circle but on the exterior of the Clifford diamond will require some non-Clifford component in their unraveling, implying $\bar c > 0$.
A sketch of this geometric structure is presented in \cref{fig:clifford_diamond}.

\begin{figure}
	\centering
	\includegraphics{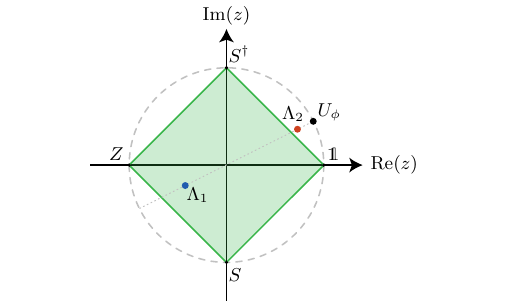}
	\caption{Geometric representation of the composite channel $\Lambda$ for aligned dephasing.
		Every rotation about the $z$ axis is a point of the unit circle (dashed): the gate $U_\phi$ sits at angle $2\phi$, and dephasing shrinks it radially by $f = 1 - 2p$.
		The four Clifford rotations $\mathcal G = \{\mathds{1}, S, \sz, S^\dagger\}$ are the vertices of the shaded diamond, whose interior collects their convex combinations, that is, all the channels that admit a purely Clifford unraveling.
		A channel falling inside the diamond, such as $\Lambda_1$, therefore has $\bar c = 0$, whereas one falling outside it, such as $\Lambda_2$, requires a non-Clifford Kraus operator and has $\bar c > 0$.
	}
	\label{fig:clifford_diamond}
\end{figure}

\subsubsection{Optimal unraveling of the channel}

\begin{figure*}[t!]
	\centering
	\includegraphics{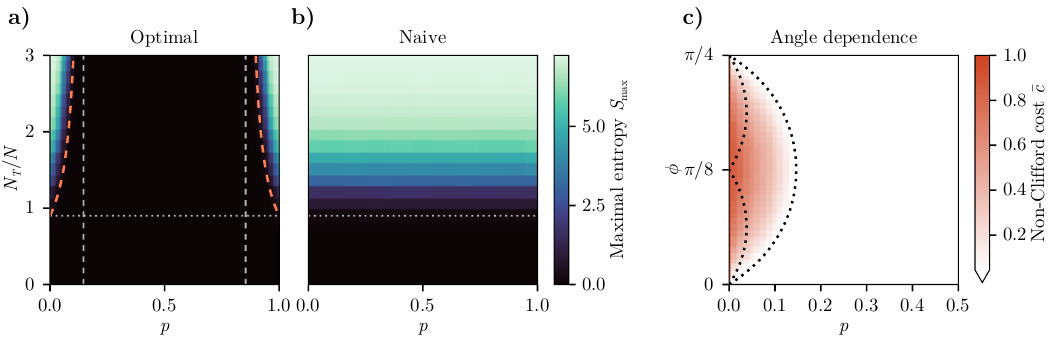}
	\caption{CAMPS simulations for aligned dephasing, with $N = 16$ qubits and $\phi = -\pi/8$.
		(a) Maximal entanglement entropy $S_{\max}$ of the inner MPS for the optimal unraveling of \cref{thm:aligned_dephasing}, as a function of the dephasing probability $p$ and of the non-Clifford density $N_T/N$.
		The orange dashed line is the predicted boundary $N_T \sim N/\bar c$, the vertical white dashed lines mark the threshold probabilities $p_\pm$ enclosing the classical phase, and the horizontal one the closed-system threshold $N_T \sim N$.
		(b) The same quantity for the naive unraveling $K_1 = T$, $K_2 = \sz T$: no classical phase appears and disentangling is lost at $N_T \sim N$ for every $p$.
		(c) Optimal cost $\bar c$ of \cref{eq:cost_aligned_dephasing} as a function of $p$ and of the rotation angle $\phi$, together with the boundaries between the three cases of the theorem (dotted lines).
		The cost is nonzero only in a lens at small $p$ around $\phi = \pi/8$, where it attains the closed-system value $\bar c = 1$ at $p = 0$.}
	\label{fig:aligned_dephasing}
	{\phantomsubcaptionlabel{fig:aligned_dephasing_optimal}}
	{\phantomsubcaptionlabel{fig:aligned_dephasing_naive}}
	{\phantomsubcaptionlabel{fig:aligned_dephasing_cost}}
\end{figure*}

We now formulate the theorem that provides the analytical expression for the cost $\bar c$ of the optimal unraveling and predicts its Kraus operators.
Given that the Clifford diamond is mapped onto itself by the application of Clifford operators, we restrict our discussion to the first quadrant $a,b \geq 0$.
Moreover, we assume $0 \leq b \leq a$ for simplicity, but an equivalent result can be derived for the other cases.
The proof of the theorem is reported in \cref{sec:proofs_aligned_dephasing}.

\begin{theorem}[Optimal unraveling for aligned dephasing]
	\label{thm:aligned_dephasing}
	Consider the aligned dephasing channel $\mathcal N_p$ in \cref{eq:aligned_dephasing} and the rotation channel $\mathcal U_\phi(\rho) = e^{i \phi \sz} \rho e^{-i \phi \sz}$.
	Defining $(a, b) = |f| (|\cos(2\phi)|, |\sin(2\phi)|)$, the cost of the optimal unraveling is
	\begin{equation}
		\label{eq:cost_aligned_dephasing}
		\bar c = \begin{dcases}
			0                            & (i) \text{ if } a+b \leq 1,                        \\
			\frac{a+b-1}{\sqrt{2}-1}     & (ii) \text{ else if } a + (\sqrt{2} - 1) b \leq 1, \\
			\frac{b^2 + (1-a)^2}{2(1-a)} & (iii) \text{ otherwise}.                           \\
		\end{dcases}
	\end{equation}

	\begin{itemize}
		\item [$(i)$] The optimal unraveling in the first case is
		      \begin{equation}
			      \Sigma = \alpha_0 I + \alpha_1 \mathcal S + \alpha_2 \mathcal Z + \alpha_3 \mathcal S^\dagger,
			      \label{eq:clifford_unraveling}
		      \end{equation}
		      where we denoted as $I, \mathcal S, \mathcal Z, \mathcal S^\dagger$ the channels associated to each operator in $\mathcal G$.
		      Introducing the slack $\tau = 1 - a - b$, the coefficients read
		      \begin{flalign}
			       & \begin{aligned}
				         \alpha_0 & = \max(a, 0)+\tau/2,  & \alpha_1 & = \max(-b, 0), \\
				         \alpha_2 & = \max(-a, 0)+\tau/2, & \alpha_3 & = \max(b, 0).
			         \end{aligned}
		      \end{flalign}

		\item [$(ii)$] In the second region we have an optimal unraveling in the form of \cref{eq:cost}, where $\Phi(\rho) = e^{i \frac{\pi}{8} Z} \rho e^{-i \frac{\pi}{8} Z}$ is a $T$-gate and $\Sigma$ is given by \cref{eq:clifford_unraveling} with the substitution
		      \begin{equation}
			      a\to\frac{a-\bar c/\sqrt{2}}{1-\bar c},\qquad b\to\frac{b-\bar c/\sqrt{2}}{1-\bar c}.
		      \end{equation}

		\item [$(iii)$] In the last region we also find an optimal unraveling as in \cref{eq:cost}.
		      The Clifford part is $\Sigma = I$, while the non-Clifford operator is a rotation about the $z$ axis by an angle satisfying $\tan(2\theta) = b /\sqrt{\bar c^2 - b^2}$.
	\end{itemize}
\end{theorem}

That is, apart from the case where the noisy channel $\Lambda$ directly lies inside the Clifford diamond, it is always possible to find an optimal unraveling that is a mixture between some Clifford channel and a non-Clifford rotation about $z$.
Interestingly, there are two regimes for the optimal unraveling depending on the rotation angle $\phi$ and the dephasing probability $p$.
In the first, one finds the optimum by placing the non-Clifford weight on a single $T$-gate regardless of $\phi$.
In the second the Clifford component collapses on one of the four vertices of the diamond.
Observe that for a $T$-gate, $a = b = |1-2p|/\sqrt{2}$.
The condition for the second region to become relevant then reads $\sqrt 2 b = |1-2p| \leq 1$, which always holds true.
Consequently, for a $T$-gate only the two regions $(i)$ and $(ii)$ are relevant to find the optimal unraveling.

\subsubsection{CAMPS simulations for aligned dephasing}

Having derived the optimal unraveling for the circuit in \cref{fig:sketch_circuit} assuming $\mathcal N$ to represent aligned dephasing, we now report numerical results proving that CAMPS can disentangle the evolution until the predicted threshold $N_T \sim N/\bar c$.
Specifically, we simulate the evolution of $N=16$ qubits in a circuit where the non-Clifford rotation is a $T$ gate ($\phi=-\pi/8$).
We monitor the maximal Entanglement Entropy (EE) across the system, defined as
\begin{equation}
	\label{eq:entanglement_entropy}
	S_{\max} = \max_j S_j, \quad S_j = -\Tr(\rho_j\log_2\rho_j),
\end{equation}
where $\rho_j = \Tr_R(\ketbra{\psi})$ is the reduced density matrix for qubits $\{1, \ldots, j\}$ and $\Tr_R$ denotes the partial trace over the remaining sites.

\cref{fig:aligned_dephasing_optimal} reports the value of $S_{\max}$ as a function of the density of non-Clifford operators $N_T/N$ and the dephasing probability $p$.
The orange dashed line represents the predicted phase boundary stemming from \cref{eq:cost_aligned_dephasing}, while the white dashed ones mark the threshold probabilities delimiting the noise-induced classical phase.

To show that the emerging classicality depends on the unraveling,
in \cref{fig:aligned_dephasing_naive} we report the value of $S_{\max}$ in the case of a \emph{naive} unraveling defined by the two Kraus operators $K_1 = T$ and $K_2 = ZT$.
That is, we simulate the case where the $T$-gate is applied first, and then one of the two Kraus operators of the bare noise channel in \cref{eq:aligned_dephasing} (either $\mathds{1}$ or $\sz$) is stochastically selected.
The behavior of $S_{\max}$ shows that for such an unraveling, no classical phase appears as $p$ is varied and CAMPS disentangle up to $N_T \sim N$, as predicted by the closed-system result.
Notice that the fact that the naive unraveling can at all disentangle up to $N_T \sim N$ should not come as a surprise, as its Kraus operators $\mathds{1}$ and $\sz$ are both Cliffords and thus can be applied to the state at zero cost.

We remark that the emergent classical phase characterizes the CAMPS and \emph{not} the MPS simulation.
Differently said, the state $\ket{\Psi}$ is not low-entangled, yet all its entanglement is of stabilizer type.
In \cref{sec:mps} we numerically show that indeed the same simulation with a simple MPS Ansatz requires exponential resources even in the classical phase.

Finally, in \cref{fig:aligned_dephasing_cost} we report the cost $\bar c$ estimated numerically from the optimal unraveling as a function of the gate angle $\phi$: each trajectory is evolved independently until Clifford operators cannot disentangle the state anymore.
The resulting layer index is used to compute the mean cost.
The rightmost black lines indicate the boundary between the classical phase with $\bar c = 0$ and the one with $0 \leq \bar c \leq 1$.
The leftmost black line marks the boundary between cases $(ii)$ and $(iii)$ of the theorem.
Interestingly, for a $T$ gate with $\phi=-\pi/8$ only case $(ii)$ ever plays a role.

\subsection{Pauli noise isotropic on the equatorial plane}
\label{sec:pauli_noise}

Having proven the emergence of classical phases in open circuits under the action of aligned dephasing, we now turn to more general noise models.
Specifically, we consider a generic Pauli noise isotropic along the equatorial plane
\begin{equation}
	\mathcal N_p(\rho) = (1-p) \rho + p_\bot (\sx \rho \sx + \sy \rho \sy) + p_z \sz \rho \sz.
	\label{eq:depolarizing}
\end{equation}
Notice that a channel in this form not only encompasses the aligned dephasing case we solved in \cref{sec:aligned_dephasing}, but also depolarizing noise (recovered by setting $p_\bot = p_z = p/3$) and some forms of twirled noise (cf. \cref{sec:twirled_dephasing}).

As in the previous section, we start by computing the entries of the PTM associated to the composite channel $\Lambda$, which reads
\begin{equation}
	L_\Lambda = \begin{pmatrix}
		1 &           \\
		  & M_\Lambda
	\end{pmatrix}, \quad
	M_\Lambda = \begin{pmatrix}
		f_\bot R(2\phi) &     \\
		                & f_z
	\end{pmatrix}
\end{equation}
where $M_\Lambda$ is the Bloch matrix, $f_\bot = 1 - 2(p_\bot + p_z)$ and $f_z = 1 - 4 p_\bot$.

Since any rotation can be decomposed as the product of three rotations about two non-orthogonal axes, and conjugation by Clifford operators can arbitrarily change the rotation axes, the most general Kraus operator $K$ arising in the unraveling has a Bloch matrix
\begin{equation}
	M_K = \tilde C_1 R_z(2\theta_1) \tilde C_2 R_z(2\theta_2) \tilde C_3 R_z(2\theta_3) \tilde C_4,
\end{equation}
where $\tilde C, R_z(2\phi)\in\so3$ are, respectively, the Bloch matrices of Clifford operators and rotations about $z$ by angles $2\theta$.
In a CAMPS simulation, the resulting computational cost $c(M_K)$ -- quantified by the number of non-Clifford operations applied to the MPS -- is thus simply determined by the number of non-trivial rotations $R_z(2\theta)$ in this decomposition.

Let us denote with $\mathcal K_m$ the set of all $M_K$ of cost $m$.
The optimal cost then reads
\begin{equation}
	\bar c = \min \qty{ \sum_m p_m m : M_\Lambda \in \sum_m p_m \conv{\mathcal K_m} },
	\label{eq:optimization_general}
\end{equation}
where $p_m \geq 0$, $\sum_m p_m = 1$ and $\conv{\mathcal K_m}$ is the set of all convex combinations of elements in $\mathcal K_m$.
In principle, the solution to this optimization problem has to be sought in the entirety of $\so{3}$.

\begin{figure}[t]
	\centering
	\includegraphics[page=2]{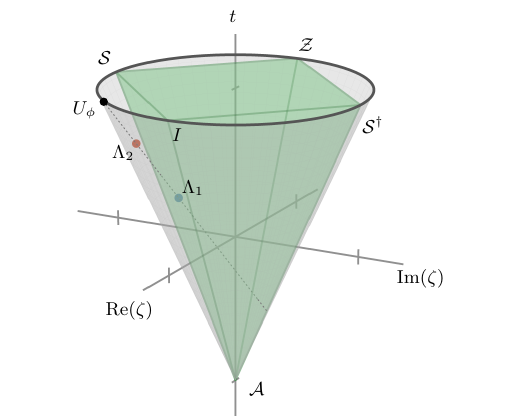}
	\caption{Geometric representation for Pauli noise isotropic on the equatorial plane.
		After the symmetrization of \cref{sec:geometric_isotropic_noise}, every channel is a point $(\zeta, t)$ of the cone $\mathcal C$ of \cref{eq:cone}: its base circle at $t = 1$ collects the rotations about the $z$ axis, among them the gate $U_\phi$, and its apex is the Clifford mixture $\mathcal A = (\mathcal X + \mathcal Y)/2$.
		The channels of zero cost fill the inscribed Clifford pyramid $\mathcal P$ of \cref{eq:pyramid} (shaded), with base vertices $I$, $\mathcal S$, $\mathcal Z$, $\mathcal S^\dagger$ and apex $\mathcal A$.
		As in \cref{fig:clifford_diamond}, $\Lambda_1$ lies inside the pyramid and has $\bar c = 0$, while $\Lambda_2$ lies outside it and has $\bar c > 0$.
	}
	\label{fig:clifford_pyramid}
\end{figure}

\subsubsection{Geometric representation of the channel}
\label{sec:geometric_isotropic_noise}
In \cref{sec:aligned_dephasing} we defined the group $\mathcal G$ of $z$ rotations that are also Clifford operators, whose Bloch matrices are $R_z(k\pi/2)$, for $k = 0,1,2,3$.
For any matrix $M \in \so{3}$, we define the \emph{symmetrization} of $M$ as
\begin{equation}
	\sym(M) = \frac{1}{4} \sum_{k = 0}^3 R_z(k\pi/2) M R_z(-k\pi/2).
\end{equation}
In \cref{sec:proofs_pauli_noise} we show that every symmetrized matrix can be parametrized by two variables only:
\begin{equation}
	\sym(M) = \begin{pmatrix}
		\Re(\zeta)  & \Im(\zeta) &   \\
		-\Im(\zeta) & \Re(\zeta) &   \\
		            &            & t
	\end{pmatrix},
\end{equation}
where $\zeta \in \mathbb{C}$, $t \in \mathbb{R}$, and $2|\zeta| = t+1$.
That is, the symmetrized matrices are represented by points on the surface of a cone in $\mathbb{C} \times \mathbb{R}$, with apex at $(0, -1)$ and base the unit circle in the plane $t = 1$.

The advantage of this symmetrization procedure is that, since $M_\Lambda$ commutes with $R_z(k\pi/2)$, symmetrizing an unraveling of $\Lambda$ produces another one of no larger cost.
The search for the optimal unraveling can thus be restricted to symmetrized matrices without loss of generality, as we show in \cref{sec:proofs_pauli_noise}.
Differently said, \cref{eq:optimization_general} is now recast into an optimization problem defined on the three-dimensional cone
\begin{equation}
	\label{eq:cone}
	\mathcal C := \conv(\sym(\so{3})) = \qty{(\zeta, t) : |\zeta| \leq \frac{1+t}{2}},
\end{equation}
where $-1\le t \le1$.
Computing $\conv(\sym(\mathcal K_0))$, it is straightforward to verify that the corresponding matrices satisfy
\begin{equation}
	\label{eq:pyramid}
	\mathcal P := \conv(\sym(\mathcal K_0)) = \qty{(\zeta, t) : ||\zeta||_1 \leq \frac{1+t}{2}},
\end{equation}
where $-1\le t \leq 1$ and $||\cdot||_1$ denotes the $\ell_1$ norm.
\cref{eq:pyramid} defines the Clifford pyramid with apex $(0, -1)$ and base vertices $(i^k, 1)$.
\cref{eq:cone} and \cref{eq:pyramid} generalize the concepts of unit circle and Clifford diamond introduced in \cref{sec:geometric_aligned_dephasing} for Pauli noise isotropic along the equator of the Bloch sphere.
A sketch of this geometric structure is presented in \cref{fig:clifford_pyramid}.

\subsubsection{Optimal unraveling of the channel}

Given that every point in $\mathcal P$ is a convex combination of Clifford channels, there exists an unraveling for which CAMPS simulations have zero cost, $\bar c=0$.
Outside the pyramid, $\bar c > 0$.
Furthermore, direct computation shows that any $\so3$ rotation $R_z(\theta)$ gets mapped onto $(\zeta, t) = (e^{i\theta}, 1)$, while the Clifford $R_x(\pi)$ is mapped to the apex $(0, -1)$.
The former costs a single non-Clifford operation and the latter none, so that the two families have cost at most one.
Moreover, they span the whole cone $\mathcal C$, thus implying that any $xy$-symmetric Pauli noise channel $\Lambda$ can be written as a mixture of a cost-zero Clifford channel $\Sigma$ and a cost-one channel $\Phi$.
We therefore recovered the general expression in \cref{eq:cost} for the noise channel described by \cref{eq:depolarizing}.
Since the time $N_T$ at which CAMPS cannot be disentangled anymore scales as $N_T \sim N / \bar c$, the bound $\bar c \leq 1$ implies that the considered noise can never hinder the classical simulation of the circuit under consideration compared to the closed-system result.

We are now in the position to state the theorem that provides the optimal cost and unraveling for the Pauli noise channel isotropic on the equator described by $(\zeta, t) = (A + i B, \, 2s_\Lambda - 1)$.
As in the previous section, we assume $0 \leq B \leq A$ for simplicity, but equivalent expressions can be found when symmetric relations hold.
The results follow directly from those presented for the aligned dephasing channel, as further discussed in \cref{sec:proofs_pauli_noise}.

\begin{figure}[t!]
	\centering
	\includegraphics{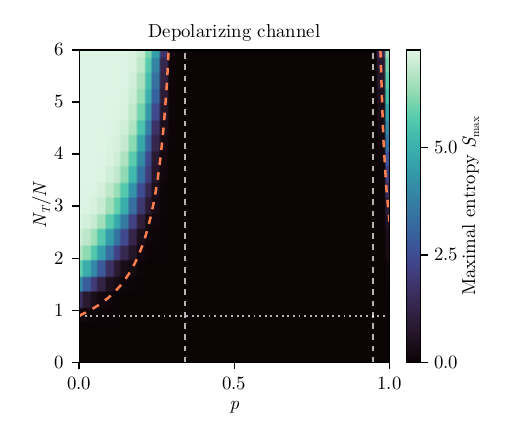}
	\caption{
		CAMPS simulations for depolarizing noise, $p_\bot = p_z = p/3$, for the same circuit as in \cref{fig:aligned_dephasing}.
		Maximal entanglement entropy $S_{\max}$ of the inner MPS as a function of $p$ and $N_T/N$, with the boundary $N_T \sim N/\bar c$ predicted by \cref{thm:pauli_noise} (orange dashed line), the threshold probabilities enclosing the classical phase (vertical white dashed lines) and the closed-system threshold $N_T \sim N$ (horizontal dotted line).
		Differently from aligned dephasing, the two regions of nonzero cost are asymmetric, since no Clifford operator maps $\mathcal N_p$ onto $\mathcal N_{1-p}$.
	}
	\label{fig:depolarizing}
\end{figure}

\begin{theorem}[Optimal unraveling for $xy$-isotropic Pauli noise]
	\label{thm:pauli_noise}
	Consider the $xy$-isotropic Pauli noise channel $\mathcal N_p$ in \cref{eq:depolarizing} and the rotation channel  $\mathcal U_\phi(\rho) = e^{i \phi \sz} \rho e^{-i \phi \sz}$.
	Define $s_\Lambda = \frac{1}{2}(1+f_z)$, $A = |f_\bot \cos(2\phi)|$, and $B = |f_\bot \sin(2\phi)|$.
	The cost of the optimal unraveling is
	\begin{equation}
		\label{eq:cost_pauli_noise}
		\bar c = \begin{dcases}
			0                                            & (i) \text{ if } A+B \leq s_\Lambda,                        \\
			\frac{A+B-s_\Lambda}{\sqrt{2}-1}             & (ii) \text{ else if } A + (\sqrt{2} - 1) B \leq s_\Lambda, \\
			\frac{B^2 + (s_\Lambda-A)^2}{2(s_\Lambda-A)} & (iii) \text{ otherwise}.                                   \\
		\end{dcases}
	\end{equation}

	\begin{itemize}
		\item [$(i)$] The optimal unraveling in the first case is
		      \begin{equation}
			      \Sigma = \alpha_0 I + \alpha_1 \mathcal S + \alpha_2 \mathcal Z + \alpha_3 \mathcal S^\dagger + (1-s_\Lambda) \mathcal A,
			      \label{eq:clifford_unraveling_pyramid}
		      \end{equation}
		      where we defined $\mathcal A(\rho) = (\sx \rho \sx + \sy \rho \sy)/2$.
		      Introducing the slack $\tau = s_\Lambda - A - B$, the coefficients read
		      \begin{flalign}
			       & \begin{aligned}
				         \alpha_0 & = \max(A, 0)+\tau/2,  & \alpha_1 & = \max(-B, 0), \\
				         \alpha_2 & = \max(-A, 0)+\tau/2, & \alpha_3 & = \max(B, 0).
			         \end{aligned}
		      \end{flalign}

		\item [$(ii)$] In the second region we have an optimal unraveling in the form of \cref{eq:cost}, where $\Phi(\rho) = e^{i \frac{\pi}{8} Z} \rho e^{-i \frac{\pi}{8} Z}$ is a $T$ gate and $\Sigma$ is given by \cref{eq:clifford_unraveling_pyramid} with the substitution
		      \begin{equation}
			      \begin{gathered}
				      A\to \frac{A -\bar c/\sqrt{2}}{1-\bar c} ,\qquad B\to \frac{B -\bar c/\sqrt{2}}{1 -\bar c}, \\
				      s_\Lambda \to \frac{s_\Lambda - \bar c}{1 - \bar c}.
			      \end{gathered}
		      \end{equation}

		\item [$(iii)$] In the last region we also find an optimal unraveling as in \cref{eq:cost}.
		      The Clifford part is
		      \begin{equation}
			      \Sigma = \frac{1 - s_\Lambda}{1-\bar c} \mathcal A + \frac{s_\Lambda - \bar c}{1 - \bar c} I,
		      \end{equation}
		      while the non-Clifford operator is a rotation about the $z$ axis by an angle satisfying $\tan(2\theta) = B /\sqrt{\bar c^2 - B^2}$.
	\end{itemize}
\end{theorem}

The structure of \cref{thm:pauli_noise} closely mirrors that of \cref{thm:aligned_dephasing} from the previous section.
This underscores the flexibility of our geometric construction for noise models that are more general than the aligned dephasing discussed in \cref{sec:aligned_dephasing}.
One such example is depolarizing noise and, more generally, any instance of Pauli noise with equal weights on the $x$ and $y$ components.

\subsubsection{CAMPS simulations for depolarizing noise}

In close analogy to the analysis presented in the previous section, we now report in \cref{fig:depolarizing} the evolution of the maximal entanglement entropy $S_{\max}$ as a function of $p$ and $N_T/N$, for a depolarizing channel as in \cref{eq:depolarizing} with $p_\bot = p_z = p/3$.

As in the previous case, we identify a range of probabilities -- marked in the plot by dashed white lines -- where the simulation remains efficient at arbitrary circuit depths.
Outside this band, the optimal unraveling renormalizes the efficiency threshold of CAMPS toward higher values, consistent with both our numerical findings and the theoretical prediction (orange dashed line).
Unlike the case of aligned dephasing, the two regions that cannot be fully disentangled are asymmetric, consequence of the fact that no Clifford operator $C$ satisfies $\mathcal N_p = C \mathcal N_{1-p} C^\dagger$.

\subsection{Tilted dephasing}
\label{sec:tilted_dephasing}

The noise family described by the channel in \cref{eq:depolarizing} allows for an unraveling which leads to the emergence of noise-induced classical phases.
Here we show that the existence of an optimal unraveling does not necessarily imply a classical phase.
On the contrary, noise can also hinder the simulability of our doped random Clifford circuit.
To make this statement explicit, in this section we study the most general form of single-qubit dephasing channel
\begin{equation}
	\mathcal{N}_{p,\vn}(\rho) = (1-p)\rho + p(\vn\cdot\vsigma)\rho(\vn\cdot\vsigma),
\end{equation}
where $\vsigma = (X,\,Y,\,Z)$ and
\begin{equation}
	\vn=(\sin\theta\cos\varphi,\,\sin\theta\sin\varphi,\,\cos\theta),
\end{equation}
with $\theta$ and $\varphi$ the polar and azimuthal angles of the noise axis, respectively.
At $(\theta,\varphi)=(0, 0)$ we recover the aligned dephasing studied in \cref{sec:aligned_dephasing}.

When dephasing is not aligned with the $U_\phi$ gate, the geometric picture of the cone and the Clifford pyramid inscribed does not apply anymore.
However, as we show in \cref{sec:proofs_tilted_dephasing}, the properties of the composite channel $\Lambda$ still lend themselves to a simple geometric interpretation.
The reason is rooted into the rigidity of $\Lambda$, whose Bloch matrix must satisfy $M_\Lambda R_z(-2\phi) \vn = \vn$.
As a consequence, each Kraus operator of a valid unraveling of $\Lambda$ can be fully parametrized by a single angle $\psi$ via
\begin{equation}
	M_K(\psi) = R_{\vn}(\psi) R_z(2\phi); \quad \psi \in [0, 2\pi).
\end{equation}
In terms of this angle, the task of finding the optimal unraveling can be recast as an optimization problem defined for a probability measure $\mu$ in the circle:
\begin{equation}
	\label{eq:optimization_tilted}
	\begin{split}
		\bar{c} = & \min_\mu \int \dd \mu(\psi) \, c(M_K(\psi)) \text{ such that} \\
		          & (i) \, \int \dd \mu(\psi) \qty(1 - \cos(\psi)) = 2p,          \\
		          & (ii) \, \int \dd \mu(\psi) \sin(\psi) = 0,                    \\
		          & (iii) \, \int \dd \mu(\psi) = 1.
	\end{split}
\end{equation}
We refer the reader to the appendix for further details and the proof.
We do not provide a closed analytical solution for the problem in \cref{eq:optimization_tilted}, but we present numerical results showcasing the richer phenomenology this model endows.

\begin{figure}
	\centering
	\includegraphics{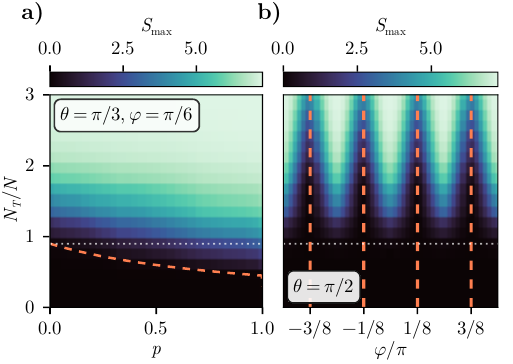}
	\caption{
		CAMPS simulations for tilted dephasing with $\phi = -\pi/8$, using the unraveling obtained by solving \cref{eq:optimization_tilted} numerically.
		(a) Maximal entanglement entropy $S_{\max}$ for the generic noise axis $(\theta, \varphi) = (\pi/3, \pi/6)$, as a function of $p$ and $N_T/N$: noise now hinders the simulation, and the predicted boundary (orange dashed line) falls below the closed-system threshold $N_T \sim N$ (horizontal white dotted line).
		(b) $S_{\max}$ for a noise axis lying on the equatorial plane, $\theta = \pi/2$, as a function of its azimuth $\varphi$: every equatorial axis remains disentanglable at least up to $N_T \sim N$, and the effect is maximal at $\varphi = \pm\pi/8, \pm3\pi/8$ (vertical dotted lines), where $\vn\cdot\vsigma$ and $U_\phi$ compose into a single Clifford operation.
		The dephasing probability is $p=50\%$.
	}
	\label{fig:tilted_dephasing}
	{\phantomsubcaptionlabel{fig:tilted_dephasing_generic}}
	{\phantomsubcaptionlabel{fig:tilted_dephasing_equator}}
\end{figure}

\subsubsection{CAMPS simulations for tilted dephasing}

For selected polar and azimuthal angles $(\theta, \varphi)$ and $\phi = -\pi/8$ (a $T$ gate), we numerically solve the optimization problem in \cref{eq:optimization_tilted} using the JuMP Julia package and the HiGHS algorithm in particular~\cite{lubin2023jump, huangfu2017parallelizing}.
Once the optimal unraveling is found, we simulate the evolution of the noisy quantum circuit using the CAMPS Ansatz and monitor the maximal entanglement entropy.

In \cref{fig:tilted_dephasing_generic} we report $S_{\max}$ for a generic angle $(\theta, \varphi) = (\pi/3, \pi/6)$ as a function of $p$ and $N_T/N$.
As visible, the effect of noise now hinders the classical simulability of the circuit, which can now be disentangled up to $N_T \lesssim N$.
Intuitively, this tilted dephasing noise, in conjunction with $U_\phi$ and the global Clifford unitaries, is injecting non-Clifford entanglement at rate larger than what CAMPS can handle.

Not all noise axes $\vn$ obstruct simulability.
In addition to the aligned dephasing case $\vn = (0, 0, 1)$, in \cref{fig:tilted_dephasing_equator} we show that any axis on the equatorial plane can be disentangled at least up to $N_T \sim N$.
In particular, this phenomenon is maximal for $\varphi = \pm \pi / 8, \pm 3\pi / 8$.
There the Kraus operator $\hat{n}\cdot\sigma$ -- itself a non-Clifford $\pi$-rotation -- combines with $U_\phi$ into a single Clifford operation.

Finally, in \cref{fig:twirled_dephasing_before} we display the optimal cost $\bar c$ found numerically as a function of the dephasing probability $p$, for various pairs $(\theta, \varphi)$.
Only the aligned dephasing case leads to the appearance of a classical phase (extended region with $\bar c = 0$), whereas other choices of $(\theta, \varphi)$ renormalize $\bar c$ with respect to the closed-system result~\cite{fux2025disentangling}.

\begin{figure}
	\centering
	\includegraphics[width=\linewidth]{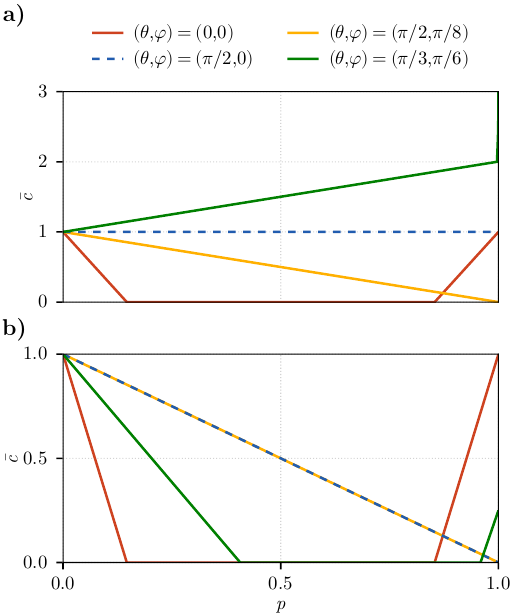}
	\caption{Optimal cost $\bar c$ of the tilted dephasing channel as a function of the dephasing probability $p$, for the noise axes $(\theta, \varphi)$ of the legend and $\phi = -\pi/8$; note the different vertical scales.
		(a) Before twirling: only the aligned axis $(0,0)$ develops a classical phase, an extended interval where $\bar c = 0$, whereas the generic axis $(\pi/3, \pi/6)$ has $\bar c > 1$ and is therefore harder to simulate than the closed-system circuit.
		(b) After twirling with $\mathcal G$: the cost is everywhere smaller than or equal to its untwirled counterpart, as proven in \cref{lemma:twirling}, and the generic axis acquires a classical phase of its own.
		The azimuth drops out of the twirled channel, cf. \cref{app_eq:twirled_channel}, so that the two equatorial axes $\varphi = 0$ and $\varphi = \pi/8$ now coincide.}
	\label{fig:twirled_dephasing}
	{\phantomsubcaptionlabel{fig:twirled_dephasing_before}}
	{\phantomsubcaptionlabel{fig:twirled_dephasing_after}}
\end{figure}

\subsubsection{Clifford-twirling the noise}
\label{sec:twirled_dephasing}

We showed previously that the tilted dephasing noise channel $\mathcal N_{p,\vn}$ makes CAMPS-based simulations inefficient for generic angles.
The same situation is a problem for experimental settings, as coherent errors
can interfere constructively, so that the intended coherent evolution is swamped.
The standard error mitigation technique is \emph{twirling}, which converts complex coherent errors into predictable stochastic noise \cite{bennett1996,wallman2016}.
For a channel $\Lambda$ and a given set of operators $G$, the $G$-twirled channel reads
\begin{equation}
	\label{eq:twirling}
	\Lambda^{\rm{tw}}(\rho) = \frac{1}{|G|}\sum_{g \in G} g^\dagger \Lambda(g\rho g^\dagger) g,
\end{equation}
where $|G|$ is the cardinality of the set.
Operationally, \cref{eq:twirling} corresponds to inserting a unitary $g \in G$ chosen uniformly at random before the channel, followed by its inverse $g^\dagger$ immediately after.
Importantly, twirling \emph{is not} a way of rewriting a channel.
It changes the evolution actually implemented by the noise, while preserving its legitimate coherent part.

As we show in \cref{sec:proofs_tilted_dephasing}, twirling the channel $\Lambda$ with $\mathcal G = \{\mathds{1}, S, Z, S^\dagger\}$ as the twirling set, both preserves the coherent evolution of $U_\phi$ and maps the noise channel to
\begin{equation}
	\mathcal N_{p, \vn}^{\rm tw} = (1-p)\,I + p\frac{\sin^2(\theta)}{2} \,\qty(\mathcal X + \mathcal Y) + p\cos^2(\theta) \, \mathcal Z,
\end{equation}
that is, it maps $\mathcal N_{p, \vn}$ into a Pauli channel isotropic along the equatorial plane.
Importantly, this is exactly the general form in \cref{eq:depolarizing} which we discussed in \cref{sec:pauli_noise}.
Therefore, all derivations and results for the optimal cost $\bar c$ and unraveling apply to twirled dephasing verbatim.

In \cref{fig:twirled_dephasing_after} we display the optimal cost of the twirled tilted dephasing channel for the same combinations of $(\theta, \varphi)$ as in \cref{fig:twirled_dephasing_before}.
We observe that $\bar c$ is always lower or at most equal with respect to the non-twirled counterpart.
Moreover, twirling leads to the emergence of new classical phases characterized by $\bar c = 0$ at angles $(\theta,\varphi)$ that were previously inaccessible to classical computation, for which CAMPS-based simulations are efficient at arbitrary circuit depth.

\section{Nonstabilizerness measures}
\label{sec:magic}

The emergence of classically simulable phases in noisy random quantum circuits is deeply connected to the amount of entanglement that cannot be removed by Clifford transformations.
A natural question is how our findings can be recast in the more general language of quantum resource theories.
The paradigmatic resource in this context is nonstabilizerness (commonly known as \emph{magic}), which quantifies the deviation of a quantum state from the set of stabilizer states~\cite{bravyi2005universal, veitch2014resource}.
In this section we show that the optimal cost $\bar c$ is not merely a property determined by the choice of a specific Ansatz, but is lower bounded by the nonstabilizerness of the noisy gate itself.
Conversely, we show for a specific measure of nonstabilizerness that, alone, it cannot resolve the boundary between the disentanglable and the genuinely quantum phase, so that the two resources are complementary diagnostics.

\subsection{Robustness of magic}

\begin{figure*}[t!]
	\centering
	\includegraphics{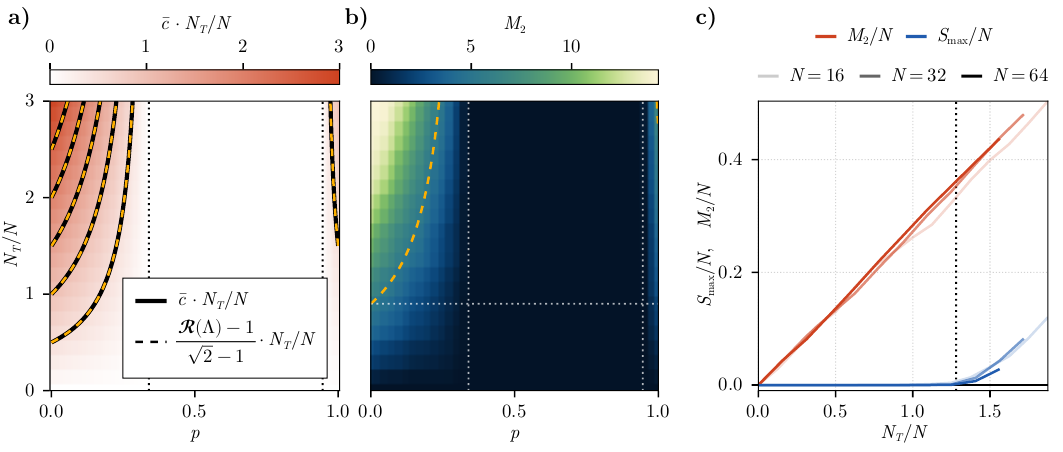}
	\caption{Nonstabilizerness of the noisy gate and of its trajectories, for the circuit of \cref{fig:sketch_circuit} with $\phi = -\pi/8$ and depolarizing noise.
		(a) Optimal cost accumulated up to depth $N_T$, that is $\bar c \, N_T/N$, as a function of $p$ and $N_T/N$.
		The solid lines are its level curves and the dashed ones the contours of the rescaled robustness of magic $(\mathcal R(J_\Lambda) - 1)/(\sqrt 2 - 1) \cdot N_T/N$: the two families coincide, as predicted by \cref{thm:magic}.
		(b) Trajectory-averaged stabilizer R\'enyi entropy $M_2$ over the same plane, with the predicted boundary $N_T \sim N/\bar c$ (dashed line).
		$M_2$ vanishes in the classical phase, but is blind to the boundary between the disentanglable and the genuinely quantum regions that \cref{fig:depolarizing} resolves.
		(c) Cut of the phase diagram at $p = 0.1$ for increasing system size.
		Once rescaled to intensive quantities the curves collapse: $M_2/N$ grows linearly from the first layer, whereas $S_{\max}/N$ departs from zero only at the predicted threshold $N_T \sim N/\bar c$ (vertical dotted line).
	}
	\label{fig:magic}
\end{figure*}

The \emph{Robustness of Magic} (RoM) associated to a state $\rho$, first introduced in Ref.~\cite{howard2017application}, is defined as
\begin{equation}
	\mathcal R(\rho) = \min \qty{ \sum_k |x_k| : \rho = \sum_k x_k \sigma_k },
	\label{eq:robustness_of_magic}
\end{equation}
where the $\sigma_k$ run over the set of pure stabilizer states of $N$ qubits, and $x_k$ real coefficients such that $\sum_k x_k=1$.
Intuitively, $\mathcal R(\rho) = 1$ if and only if $\rho$ admits a decomposition with no negative weight, i.e., if and only if $\rho$ is a convex combination of pure stabilizer states.
Conversely, states that don't admit such a decomposition, will necessarily require negative $x_k$, and the more negativity is needed, the larger $\mathcal R(\rho) > 1$.
This is one of the properties that make the RoM a faithful magic monotone: it is non-increasing under stabilizer operations and under classical randomness~\cite{howard2017application}.

In analogy to quantum states, the RoM can be extended to a quantum channel $\Lambda$ via the Choi–Jamio\l{}kowski isomorphism, which maps $\Lambda$ onto the state
\begin{equation}
	J_\Lambda = (\Lambda \otimes \mathds{1}) \ketbra{\Phi}{\Phi},
\end{equation}
where $\ket{\Phi}$ is a normalized $2$-qubit maximally entangled state.
If $\Lambda$ is a probabilistic mixture of stabilizer operations, then $J_\Lambda$ is a convex combination of stabilizer states, and $\mathcal R(J_\Lambda) = 1$.
Consequently, $\mathcal R(J_\Lambda) > 1$ certifies that $\Lambda$ contains a genuinely non-stabilizer component.
For the unital qubit channels of interest here ($\Lambda(\mathds{1}) = \mathds{1}$), the converse also holds, so that $\mathcal R(J_\Lambda)$ faithfully detects the nonstabilizerness of the channel.
Importantly, since all channel unravelings share the same Choi state, $\mathcal{R}(J_\Lambda)$ does not depend on the unraveling.
Indeed, the minimization in \cref{eq:robustness_of_magic} is over signed decompositions of a fixed geometric object, not over physical realizations.
Notably, for some of the noise models we considered in \cref{sec:classical_phases}, $\mathcal{R}(J_\Lambda)$ can be evaluated in closed form.
For aligned dephasing (cf. \cref{eq:aligned_dephasing}) and depolarizing noise (cf. \cref{eq:depolarizing} with $p_\bot=p_z=p/3$) we find the following expressions, respectively~\cite{howard2017application,seddon2019quantifying}
\begin{equation}
	\begin{split}
		 & \mathcal{R}(J_\Lambda) = \max\left(1, \sqrt{2}|1-2p|\right),                                         \\
		 & \mathcal{R}(J_\Lambda) = \max\left(1, \frac{1+(2\sqrt{2}-1)q}{2}, \frac{1-(2\sqrt{2}+1)q}{2}\right),
	\end{split}
	\label{eqs:robustness_closed}
\end{equation}
where $q = 1-\frac{4}{3}p$.
A thorough derivation of \cref{eqs:robustness_closed} can be found in \cref{sec:proofs_magic}.
The kink in the closed expressions of $\mathcal{R}(J_\Lambda)$ delimits the boundary of the noise-induced classical phase.
The non-analyticity of $\mathcal{R}(J_\Lambda)$ is rooted in the geometric structure of the channel $\Lambda$: inside the Clifford diamond/pyramid, the channel admits a decomposition in stabilizer operations only; outside a non-Clifford operation is required.
The polytope facet is the sharp boundary between the two scenarios, and reflects in the kink of $\mathcal{R}(J_\Lambda)$.
A more in depth discussion of the geometric interpretation of $\mathcal{R}(J_\Lambda)$, which complements the one outlined in \cref{sec:geometric_aligned_dephasing} and \cref{sec:geometric_isotropic_noise}, can be found in the \cref{sec:proofs_magic}.

The main result we derive in this section is that $\mathcal R(J_\Lambda)$ can be formally connected to the optimal cost $\bar c$ defined and discussed extensively in the previous sections, even though $\mathcal R(J_\Lambda)$ is a property of the channel alone, while $\bar c$ is defined operationally by optimizing over the unravelings of $\Lambda$.
We have the following theorem:
\begin{theorem}[Relation between $\bar c$ and RoM]
	\label{thm:magic}
	For any unital qubit channel $\Lambda$, denote with $J_\Lambda$ the associated Choi state.
	The non-Clifford fraction $\bar c$, that is, the cost of an unraveling for a CAMPS-based Ansatz, satisfies
	\begin{equation}
		\label{eq:magic_bound}
		\bar c \geq \frac{3(\mathcal R(J_\Lambda) - 1)}{2\sqrt 2 - 1}.
	\end{equation}
	Moreover, for the circuit in \cref{fig:sketch} with $\phi = -\pi/8$ and a Pauli noise channel isotropic in the equatorial plane, the cost is exactly
	\begin{equation}
		\label{eq:relation_cost_rom}
		\bar c = \frac{\mathcal R(J_\Lambda) - 1}{\sqrt 2 - 1}.
	\end{equation}
\end{theorem}
The proof of this result is reported in \cref{sec:proofs_magic}.
We want to remark a striking fact of \cref{eq:magic_bound}.
The cost $\bar c$ is a quantity defined operationally: it is the smallest density of non-Clifford operations attainable over all the ways of unraveling $\Lambda$ into quantum trajectories with unitary Kraus operators.
The right-hand side is instead a fixed property of the channel, extracted from an optimization problem over signed decompositions of $J_\Lambda$.
A priori, there is no reason for the two to be related, since infinitely many inequivalent unravelings with widely different non-Clifford content share the same Choi state.
What the theorem states is that the freedom in the unraveling cannot be exploited indefinitely, and no choice of unraveling can push the non-Clifford cost below the magic that the channel inherently carries as a whole.
From this perspective, the classically simulable phases identified in the previous sections admit a resource-theoretic interpretation as the regions in which noise degrades the nonstabilizerness of the channel faster than the circuit injects it.

We report in \cref{fig:magic} the optimal cost rescaled by non-Clifford density $\bar c \cdot N_T/N $, estimated numerically from quantum trajectories for a circuit with $T$ gates and depolarizing noise, together with the rescaled robustness $(\mathcal R(J_\Lambda) - 1)/(\sqrt 2 - 1)$ evaluated from \cref{eqs:robustness_closed}.
The heatmap refers to the optimal cost, with the black lines marking its level curves.
The yellow-dashed lines correspond to a contour plot to the RoM, suitably rescaled as in \cref{eq:relation_cost_rom}.
As predicted by \cref{thm:magic}, the two perfectly agree.

\subsection{Stabilizer Renyi entropy}

A second, possibly more widespread, measure of nonstabilizerness is the \emph{stabilizer Renyi entropy} (SRE)~\cite{leone2022stabilizer}, defined on a pure state $\ket{\Psi}$ as
\begin{equation}
	M_2(\ket{\Psi}) = -\log(\frac{1}{2^N} \sum_{P \in \mathcal P_N} \qty|\bra{\Psi} P \ket{\Psi}|^4),
\end{equation}
where $\mathcal P_N$ is the Pauli group of $N$ qubits.
$M_2$ measures how spread out the state is over the Pauli basis.
A stabilizer state has a uniform distribution of $\qty|\bra{\Psi} P \ket{\Psi}|^2$, giving $M_2 = 0$, while any support beyond a stabilizer group yields $M_2 > 0$.
Being defined on pure states, $M_2$ is naturally suited to the study of individual quantum trajectories rather than of the full channel.

In \cref{fig:magic} we report the trajectory-averaged $M_2$ for the same circuit subject to depolarizing noise.
The comparison with the EE of \cref{fig:depolarizing} highlights that, while $M_2$ does detect the fully classical phase, it is blind to the boundary between the classically disentanglable phase and the genuinely quantum one at $\bar c > 0$.
Taken together, the two diagnostics identify three regions of parameters (cf. \cref{fig:sketch_phases}):
\begin{enumerate}
	\item[I.] A disentanglable region, marked by $S_{\rm max} = 0$ and $M_2 > 0$.
	      Nonstabilizerness is present, yet it can be entirely absorbed into the Clifford part of the Ansatz;
	\item[II.] A fully classical phase which is unambiguously identified by both metrics (both of them vanish);
	\item[III.] A genuinely quantum phase, entered outside the classical region once the circuit is deep enough.
	      Intuitively, this can be understood in light of the fact that, even though noise reduces the rate at which magic is injected, the closed-system behaviour of Ref.~\cite{fux2025disentangling} is recovered at large enough $N_T$.
	      Then, both $S_{\rm max}$ and $M_2$ grow and the resources required by the Ansatz increase exponentially.
\end{enumerate}
The existence of the first region proves why the SRE alone cannot be used as a diagnostic for classical simulability within CAMPS.
What matters is not how much magic the state contains, but whether that magic can be reabsorbed by a Clifford transformation.

Finally, in \cref{fig:magic} we report data along a cut of the phase diagram at $p = 0.1$, for different system sizes.
Once rescaled to intensive quantities, the behavior is consistently found irrespective of the system size.
As $N$ increases, the entanglement-entropy density shows that disentangling the state remains possible up to $N_T \sim N$, the only finite-size effect being that the crossover at $N_T \sim N$ is not sharp.
Similarly, $M_2$ increases linearly already from the first application of the non-Clifford operation, irrespective of the system size.

\section{Disentangling the density operator}
\label{sec:campo_disentangling}

So far we focused on the simulation of open quantum circuits via stochastic quantum trajectories.
A complementary approach consists in evolving the averaged density operator $\rho$ via the channel $\Lambda$ in \cref{eq:lambda_channel} directly.
Since the density operator is the unconditioned average of all Kraus operators, any unraveling must leave the averaged dynamics generated by \cref{eq:lambda_channel} unchanged.
That is, the optimal unraveling search, which we proved to be fundamental for the appearance of classical phases in noisy quantum circuits, loses its applicability when looking at the dynamics of $\rho$.
In this section we ask whether classical phases can nevertheless emerge for the averaged open system dynamics.
To this purpose, we generalize the CAMPS Ansatz to matrix product operators (MPOs) \cite{verstraete2004,zwolak2004}, thus defining the \emph{Clifford-augmented Matrix Product Operator} (CAMPO) Ansatz
\begin{equation}
	\rho_{\rm CAMPO} = C\,\rho_{\rm MPO}\,C^\dagger.
\end{equation}

We consider the circuit in \cref{fig:sketch_circuit} and formulate the following rigorous result valid for the unconditioned dynamics generated by $\Lambda$ under Pauli noise:

\begin{theorem}[Structure of the evolved density operator]
	\label{thm:density_operator}
	Given the channel $\Lambda$ in \cref{eq:lambda_channel} subject to Pauli noise, up to the application of $N_T \sim N$ non-Clifford operators,
	\begin{equation}
		\rho = \tilde{C}\,\Sigma\,\tilde{C}^\dagger,\qquad\Sigma = \sum_\nu q_\nu\bigotimes_{j=1}^N\sigma_j^\nu,
	\end{equation}
	where $\sigma_j^\nu$ denotes the density matrix associated to a single-qubit pure state.
	That is, the density operator is a single-Clifford conjugation of the classical mixture of product states $\Sigma$.
	Moreover, the disentangling threshold is independent of the noise strength and noise contributes only classical correlations, without affecting the entanglement of the individual states in the classical mixture $\Sigma$.
\end{theorem}

\begin{figure}
	\centering
	\includegraphics[width=\linewidth]{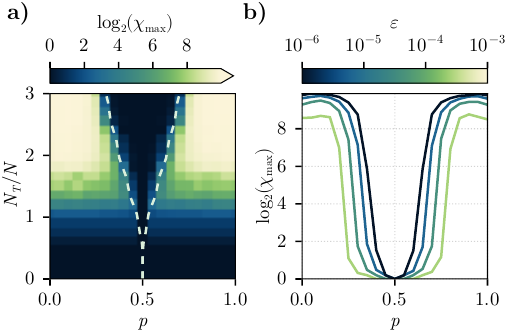}
	\caption{Disentangling the averaged density operator with the CAMPO Ansatz, for aligned dephasing.
		(a) Log-bond dimension $\log_2(\chi_{\max})$ of the inner MPO as a function of $p$ and $N_T/N$.
		Disentangling is possible only up to $N_T \sim N$, the closed-system threshold, independently of the noise strength, as predicted by \cref{thm:density_operator}.
		The region of vanishing bond dimension that opens above it around $p = 1/2$ is not due to disentangling, but to the relaxation of $\rho$ towards $\mathds{1}/2^N$: the dashed lines are the truncation thresholds $p_\pm(N_T)$ of \cref{app_eq:campo_threshold}.
		(b) $\log_2(\chi_{\max})$ as a function of $p$ for decreasing values of the MPO truncation error $\varepsilon$.
		That region shrinks as $\varepsilon \to 0$, confirming that it is an artifact of the truncation.
	}
	\label{fig:campo}
	{\phantomsubcaptionlabel{fig:campo_heatmap}}
	{\phantomsubcaptionlabel{fig:campo_scaling}}
\end{figure}

The proof of \cref{thm:density_operator} is reported in \cref{sec:proofs_campo_disentangling}.
We remark that the theorem predicts that a single Clifford operator is able to disentangle the density matrix regardless of the noise structure up to $N_T\sim N$: exactly the closed-system result~\cite{fux2025disentangling}.
Therefore, we do not expect the CAMPO Ansatz to unveil classical phases for generic noisy quantum dynamics.
Furthermore, the ability of CAMPO to disentangle dynamics up to $N_T \sim N$, does not imply that its simulation will be efficient.
In fact, the inner MPO representing $\Sigma$ will, in general, have to pay bond dimension in order to describe faithfully the classical correlations.
Whether CAMPO would or not allow for efficient simulation entirely depends on the details of the circuit being studied.
For example, the averaged dynamics generated by a circuit of Cliffords + $T$-gates with aligned dephasing approaches the infinite-temperature state $\rho \propto \mathds{1}$ at long times (cf. \cref{sec:proofs_campo_disentangling}).
This state requires bond dimension $\chi=1$, and thus a CAMPO-based Ansatz is able to accurately describe it.

In \cref{fig:campo_heatmap}, we report numerical results supporting our theoretical findings, in the case of aligned dephasing noise.
In particular, we plot the log-bond dimension $\log_2(\chi)$ as a function of $N_T/N$ and $p$, showing that CAMPO can only disentangle up to $N_T \sim N$, mimicking the closed-system result.
In particular, the region with zero log-bond dimension in the central part of the plot \emph{is not} due to a disentangling action operated by the Clifford operator in $\rho_{\rm CAMPO}$, but rather due to the steady-state density operator that only requires $\chi=1$ to be represented.
Two arguments support this claim.
First, in \cref{sec:proofs_campo_disentangling} we estimate the threshold probability at which singular values of the MPO would be truncated assuming the dynamics of the circuit under consideration.
The white dashed lines in \cref{fig:campo_heatmap} display this threshold probability and align well with the numerically-computed $\log_2(\chi)$.
Secondly, and consistently with the first argument, this region shrinks as the truncation error $\varepsilon$ of the MPO is reduced, pointing to the fact that it would disappear altogether for an exact simulation (see \cref{fig:campo_scaling}).

When compared to the data presented in \cref{sec:aligned_dephasing}, it appears evident how the classical simulability of open quantum dynamics changes qualitatively when adopting the density-matrix (averaged) or wave-function (single instance) description.
The explanation relies in the simple inequality
\begin{equation}
	\rho_{\rm CAMPO} = C\sum_{j=1}^{N_{\rm traj}}\frac{\ketbra{\Psi_j}}{N_{\rm traj}}C^\dagger\ne \sum_{j=1}^{N_{\rm traj}}\frac{C_j\ketbra{\Psi_j}C_j^\dagger}{N_{\rm traj}}.
\end{equation}
Unraveling $\Lambda$ in quantum trajectories allows each trajectory to pick \emph{its own} optimal Clifford independently, whereas the density-matrix description is restricted to a unique \emph{global} Clifford, a much more stringent condition.

\section{Discussion}
\label{sec:discussion}
We have shown that the freedom in unraveling in an open quantum circuit dynamics can be exploited to reduce the specific quantum resource that limits a chosen classical representation.
For Clifford circuits doped with noisy non-Clifford rotations, we introduced the average non-Clifford weight $\bar c$ as an operational cost adapted to CAMPS and determined the unraveling that minimizes it.
The resulting cost directly controls the depth up to which the trajectories can be completely Clifford-disentangled, $N_T\sim N/\bar c$.
Most importantly, when $\bar c=0$, the noisy channel admits a purely Clifford unraveling and the trajectories remain disentanglable at arbitrary circuit depths, giving rise to a noise-induced classical phase.

Using a geometric representation of single-qubit channels, we obtained the optimal unraveling analytically for aligned dephasing and for a broad family of Pauli noise, including depolarizing noise.
CAMPS simulations quantitatively reproduce the corresponding phase boundaries and demonstrate that the emergent classicality is intrinsically tied to the choice of unraveling: a naive Kraus decomposition of the same physical channel does not reveal the classical phase.
Conversely, noise does not universally aid classical simulation.
For tilted dephasing, the noise can increase the effective non-Clifford cost and hinder the ability to disentangle, while twirling the channel can reduce this cost and recover classically accessible regions.
Thus, it is the geometry of the noisy channel relative to the efficiently simulable manifold, rather than the noise strength alone, that determines whether dissipation simplifies the dynamics.

Our results also clarify the role of quantum resources in this transition.
The robustness of magic of the Choi state provides an unraveling-independent lower bound on the optimal cost and, for the Pauli-noise family considered here, is directly proportional to $\bar c$.
At the trajectory level, however, nonstabilizerness and residual entanglement provide complementary information.
Their interplay separates the dynamics into a disentanglable regime, where finite magic can still be accommodated without residual MPS entanglement; a classical phase, where both quantities vanish; and a quantum regime, where the non-Clifford resource can no longer be absorbed by Clifford transformations and the resources required by CAMPS grow rapidly.
This shows that neither entanglement nor nonstabilizerness alone determine classical simulability: the relevant quantity is the portion of the quantum resource that cannot be reorganized within the efficiently simulable sector of the representation.

The comparison with the unconditional dynamics further highlights the special role of quantum trajectories.
For the averaged density operator, a single Clifford frame disentangles the dynamics only up to the closed-system scale $N_T\sim N$, independently of the noise strength, and the unraveling freedom responsible for the trajectory classical phase is lost.
Optimal unravelings should therefore be viewed not merely as convenient representations of open-system dynamics, but as an additional computational degree of freedom that can expose classical structure hidden in the density-operator description.

Several directions follow naturally from these results.
An important extension is to continuous-time open dynamics governed by Lindblad master equations.
There, the freedom in choosing quantum-jump or diffusive unravelings could be optimized with respect to a resource-adapted cost, potentially defining a continuous-time rate of non-Clifford resource injection and establishing whether analogous classical phases survive in Hamiltonian many-body dynamics with dissipation.
A second direction is to replace the stabilizer manifold by other efficiently simulable structures.
For fermionic systems, Gaussian states and matchgate circuits provide the natural analogue of stabilizer states and Clifford dynamics.
Optimizing the unraveling of noisy doped-matchgate circuits with respect to the injected non-Gaussian resource could therefore reveal noise-induced Gaussian phases and determine whether irreducible fermionic non-Gaussianity plays the same role that the non-Clifford cost plays here.
More broadly, our results suggest a representation-dependent principle for open-system simulation: the useful unraveling is the one that minimizes the resource obstructing the efficiently simulable manifold on which the classical algorithm is built.

\begin{acknowledgments}
	We acknowledge useful discussions with Piotr Sierant, Vincenzo Savona, Marcello Dalmonte, Mario Collura, Alessio Paviglianiti and Alberto Mercurio.
	L.F. and F.F. acknowledge support by the Swiss National Science Foundation through Projects No. 200020\_215172, 200021-227992, and 20QU-1\_215928, and as part of NCCR SPIN (grant number 225153).
	E.T. was funded by the Swiss National Science Foundation (SNSF) under Grant No. TMPFP2\_234754. E.T. acknowledges CINECA (Consorzio Interuniversitario per il Calcolo Automatico) award, under the ISCRA initiative and Leonardo early access program, for the availability of high-performance computing resources and support.
\end{acknowledgments}

\appendix
\crefalias{section}{appendix}

\section{Proofs for aligned dephasing}
\label{sec:proofs_aligned_dephasing}

This appendix proves the results stated in \cref{sec:aligned_dephasing}.
The argument proceeds in two steps.
We first show that aligned dephasing constrains its unravelings so tightly that the search for the optimum can be transcribed into the complex plane: every admissible Kraus operator is a rotation about the $z$ axis, and the Clifford ones mix exactly into the diamond of \cref{fig:clifford_diamond} (\cref{app:aligned_reduction}).
The cost $\bar c$ from \cref{eq:cost} is then the solution of an elementary planar problem, which we solve in closed form in \cref{app:aligned_optimum}, thereby proving \cref{thm:aligned_dephasing}.

\subsection{Reduction to a planar problem}
\label{app:aligned_reduction}

\begin{lemma}[Admissible Kraus operators]
	\label{lemma:z_rotations}
	Let $\Lambda$ be a qubit channel whose Bloch matrix satisfies $(M_\Lambda)_{33} = 1$, and let $\Lambda(\rho)=\sum_j\lambda_jK_j\rho K_j^\dagger$ be any unraveling of $\Lambda$ into unitary Kraus operators.
	Then every $K_j$ is a rotation about the $z$ axis.
	Among these, the Clifford ones are exactly the four elements of $\mathcal G = \{\mathds{1}, S, \sz, S^\dagger\}$, the cyclic group of order four generated by the phase gate $S$.
\end{lemma}
\begin{proof}
	Let $M_j \in \so{3}$ be the Bloch matrix associated to the Kraus operator $K_j$, so that $(M_\Lambda)_{33} = \sum_j \lambda_j (M_j)_{33}$ by linearity.
	Every rotation obeys $|(M_j)_{33}| \leq 1$, and the weights $\lambda_j$ are non-negative and sum to one.
	Therefore, $1 = \sum_j \lambda_j (M_j)_{33} \leq \sum_j \lambda_j = 1$, which can only happen if $(M_j)_{33} = 1 \, \forall j$.
	A rotation that leaves the $z$ axis invariant is a rotation about it, that is, $M_j = R_z(2\theta_j)$ and $K_j = e^{i\theta_j \sz}$.
	Such an operator maps the Pauli group onto itself if and only if $2\theta_j$ is an integer multiple of $\pi/2$, which singles out $\mathds{1}$, $S$, $\sz$ and $S^\dagger$.
\end{proof}

The composite channel $\Lambda = \mathcal N_p \circ \mathcal U_\phi$ satisfies the hypothesis of \cref{lemma:z_rotations}, since aligned dephasing leaves the $z$ axis untouched, cf.~\cref{eq:aligned_dephasing_ptm}.
Every Kraus operator of every unraveling of $\Lambda$ is therefore a rotation about $z$, and is faithfully described by the single complex number introduced in \cref{sec:geometric_aligned_dephasing}, which encodes the $xy$ block of its Bloch matrix.
The channel to be reproduced sits at $z_\Lambda = f(\cos(2\phi) + i\sin(2\phi))$, and each $K_j$ carries a cost of zero or one, according to whether it belongs to $\mathcal G$ or not.
The next lemma identifies the set of coordinates that are reachable at zero cost.

\begin{lemma}[Clifford mixtures]
	\label{lemma:clifford_channel}
	Let
	\begin{equation}
		\label{app_eq:clifford_channel}
		\Sigma = \alpha_0 I + \alpha_1 \mathcal S + \alpha_2 \mathcal Z + \alpha_3 \mathcal S^\dagger,
	\end{equation}
	with $\alpha_k \geq 0$ and $\sum_k \alpha_k = 1$, be a mixture of the Clifford channels associated to $\mathcal G$.
	Its complex coordinate is $z_\Sigma = (\alpha_0 - \alpha_2) + i (\alpha_3 - \alpha_1)$ and, as the weights vary, it covers exactly the closed $\ell_1$ ball
	\begin{equation}
		\label{app_eq:diamond}
		\mathcal D = \qty{z = a + i b : ||z||_1 := |a| + |b| \leq 1},
	\end{equation}
	that is, the diamond with vertices $\mathds{1}$, $S$, $\sz$ and $S^\dagger$.
\end{lemma}
\begin{proof}
	The $xy$ blocks of the Bloch matrices relative to the operators $\mathds{1}$, $S$, $\sz$ and $S^\dagger$ are, respectively,
	\begin{flalign}
		 & \begin{aligned}
			   R_{\mathds{1}} & = \begin{pmatrix} 1 & 0\\ 0 & 1 \end{pmatrix},   & R_{S}         & = \begin{pmatrix} 0& -1\\ 1 & 0 \end{pmatrix},  \\
			   R_{Z}          & = \begin{pmatrix} -1 & 0\\ 0 & -1 \end{pmatrix}, & R_{S^\dagger} & = \begin{pmatrix} 0& 1\\ -1 & 0  \end{pmatrix},
		   \end{aligned}
	\end{flalign}
	so that, by linearity, the block associated to $\Sigma$ is the one of a coordinate $z_\Sigma = a + i b$ with $a = \alpha_0 - \alpha_2$ and $b = \alpha_3 - \alpha_1$.
	The image is contained in $\mathcal D$, because
	\begin{equation}
		||z_\Sigma||_1 = |\alpha_0 - \alpha_2| + |\alpha_3 - \alpha_1| \leq \alpha_0 + \alpha_1 + \alpha_2 + \alpha_3 = 1.
	\end{equation}
	Conversely, given $z = a + i b$ with $|a| + |b| \leq 1$, introduce the slack $\tau = 1 - |a| - |b| \geq 0$ and the weights
	\begin{flalign}
		 & \begin{aligned}
			   \alpha_0 & = \max(a, 0)+\tau/2,  & \alpha_1 & = \max(-b, 0), \\
			   \alpha_2 & = \max(-a, 0)+\tau/2, & \alpha_3 & = \max(b, 0).
		   \end{aligned}
	\end{flalign}
	They are non-negative and add up to $|a| + |b| + \tau = 1$.
	Moreover, the identity $x = \max(x, 0) - \max(-x, 0)$ gives
	\begin{equation}
		\begin{split}
			\alpha_0 - \alpha_2 = \max(a, 0) - \max(-a, 0) = a, \\
			\alpha_3 - \alpha_1 = \max(b, 0) - \max(-b, 0) = b,
		\end{split}
	\end{equation}
	so that every point of $\mathcal D$ is attained.
\end{proof}

Taken together, the two lemmas turn the search for the optimal unraveling into a two-dimensional geometric problem.
Group the Kraus operators of an arbitrary unraveling of $\Lambda$ according to their cost, the Clifford ones carrying a total weight $1 - \bar c$ and the remaining rotations a total weight $\bar c$, and collect each group into a single channel.
This is precisely the decomposition of \cref{eq:cost}, $\Lambda = \bar c \, \Phi + (1 - \bar c) \Sigma$, in which $\bar c$ is by construction the mean number of non-Clifford operations applied per time step.
Coordinates combine linearly, $z_\Lambda = \bar c \, z_\Phi + (1-\bar c) z_\Sigma$, and the two components are constrained differently: $\Sigma$ is admissible if and only if $z_\Sigma \in \mathcal D$ by \cref{lemma:clifford_channel}, whereas $\Phi$, being a mixture of rotations, is only bound to the unit disk $\mathcal B$ by $|z_\Phi| \leq 1$.
Absorbing the weight into the non-Clifford coordinate through $v = \bar c \, z_\Phi$, the two conditions become $|v| \leq \bar c$ and $||z_\Lambda - v||_1 \leq 1 - \bar c$.
Lowering $\bar c$ tightens the former and relaxes the latter, so that the optimum saturates $\bar c = |v|$ and
\begin{equation}
	\label{app_eq:aligned_optimization}
	\bar c = \min_{v \in \mathbb{C}} |v| \quad \text{subject to} \quad ||z_\Lambda - v||_1 + |v| \leq 1.
\end{equation}
A minimizer of \cref{app_eq:aligned_optimization} reconstructs the whole unraveling: a single non-Clifford rotation $z_\Phi = v / |v|$, selected with probability $\bar c = |v|$, together with the Clifford mixture $z_\Sigma = (z_\Lambda - v)/(1 - \bar c)$, whose weights are those of \cref{lemma:clifford_channel}.
The degenerate case $v = 0$ corresponds to an unraveling that needs no non-Clifford Kraus operator at all.
Geometrically, the constraint in \cref{app_eq:aligned_optimization} requires $z_\Lambda$ to lie in $\bar c \, \mathcal B + (1 - \bar c) \mathcal D$: the optimal cost is the smallest weight for which this body, which interpolates between the Clifford diamond at $\bar c = 0$ and the unit disk at $\bar c = 1$, has inflated enough to reach the channel.

\subsection{Proof of the optimal unraveling and cost}
\label{app:aligned_optimum}

We can now solve \cref{app_eq:aligned_optimization} in closed form.
Write $z_\Lambda = a + i b$ and assume $0 \leq b \leq a$, as in the statement of \cref{thm:aligned_dephasing}.
This is no loss of generality: composing $\Lambda$ with the Clifford $S$ multiplies its coordinate by $-i$, while conjugating $\Lambda$ with $X$ maps the coordinate to its complex conjugate.
Both operations are free of cost and map $\mathcal D$ onto itself; together they generate the eight symmetries of the diamond, whose fundamental domain is exactly the sector $0 \leq b \leq a$.
In this sector $a = |f \cos(2\phi)|$ and $b = |f \sin(2\phi)|$.

If $a + b \leq 1$ the channel lies inside the Clifford diamond, $v = 0$ is feasible and therefore optimal.
No non-Clifford Kraus operator is needed, $\bar c = 0$, and \cref{lemma:clifford_channel} returns the Clifford weights quoted in case $(i)$ of \cref{thm:aligned_dephasing}, with slack $\tau = 1 - a - b$.

Assume from now on $a + b > 1$, so that no purely Clifford unraveling exists and the minimizer is non-zero.
Note that this forces $a < 1$, since $a = 1$ would imply $b = 0$ and hence $a + b = 1$.
Two reductions bring \cref{app_eq:aligned_optimization} down to one dimension.
First, the minimizer can be taken in the rectangle $\mathcal V = \qty{v = v_1 + i v_2 : 0 \leq v_1 \leq a, \, 0 \leq v_2 \leq b}$: clipping either component to its interval does not increase $|v|$ nor $||z_\Lambda - v||_1$, hence it preserves feasibility while improving the objective.
Second, inside $\mathcal V$ the $\ell_1$ distance collapses to $||z_\Lambda - v||_1 = a + b - (v_1 + v_2)$, so that the only relevant quantities are the norms $|v|$ and $||v||_1 = v_1 + v_2$.
It is then natural to scan the problem at fixed radius $t \geq 0$.
Introducing the constrained maximum
\begin{equation}
	\label{app_eq:aligned_function_m}
	m(t) = \max \qty{v_1 + v_2 : v \in \mathcal V, \, |v| \leq t},
\end{equation}
the function
\begin{equation}
	\label{app_eq:aligned_function_G}
	G(t) = a + b - m(t) + t
\end{equation}
is the smallest value that $||z_\Lambda - v||_1 + t$ attains in the disk of radius $t$.
Whenever $G(t) \leq 1$, the maximizer of \cref{app_eq:aligned_function_m} is a feasible point of \cref{app_eq:aligned_optimization} of radius at most $t$, so that $\bar c \leq t$.
Conversely, an optimal $v$ satisfies $G(|v|) \leq ||z_\Lambda - v||_1 + |v| \leq 1$, so that $\bar c$ itself belongs to the set of radii meeting the constraint.
The cost is therefore the first radius at which $G$ drops below one,
\begin{equation}
	\label{app_eq:aligned_scalar_problem}
	\bar c = \min \qty{t \geq 0 : G(t) \leq 1}.
\end{equation}

The maximum in \cref{app_eq:aligned_function_m} is elementary.
By Cauchy-Schwarz $v_1 + v_2 \leq \sqrt2 \, |v| \leq \sqrt 2\, t$, with equality at $v = (1 + i) \, t/\sqrt2$, a point that belongs to $\mathcal V$ if and only if $t \leq \sqrt2 \, b$, having assumed $b \leq a$.
When $t > \sqrt2\, b$ the constraint $v_2 \leq b$ becomes active, and the remaining budget goes into $v_1 = \sqrt{t^2 - b^2}$.
Therefore, explicitly,
\begin{equation}
	\label{app_eq:aligned_G_explicit}
	G(t) = \begin{dcases}
		a + b - (\sqrt2 - 1)\, t & \text{if } t \leq \sqrt 2 \, b, \\
		a + t - \sqrt{t^2 - b^2} & \text{if } t \geq \sqrt 2 \, b,
	\end{dcases}
\end{equation}
where the two expressions agree at $t = \sqrt2\, b$, both equal to $a + (\sqrt2 - 1) b$.
\cref{app_eq:aligned_G_explicit} is exact on $0\le t\le|z_\Lambda|$, and this is the only
range we need.
On the first branch the restriction is automatic, since $t\le\sqrt2\,b\le\sqrt{a^2+b^2}$ whenever $b\le a$; on the second it is the condition $\sqrt{t^2-b^2}\le a$ that keeps the maximizer inside the rectangle $\mathcal V$.
There $G$ is continuous and strictly decreasing, its derivative being $1-\sqrt2<0$ on the first branch and $1-t/\sqrt{t^2-b^2}<0$ on the second; it starts above the constraint, $G(0)=a+b>1$, and at the right endpoint it has already dropped to $G(\sqrt{a^2+b^2})=\sqrt{a^2+b^2}=|z_\Lambda|\le1$.
Hence $G(t)=1$ has a unique root in $[0,\sqrt{a^2+b^2}]$, and by \cref{app_eq:aligned_scalar_problem} this root is the optimal cost $\bar c$.
The two remaining cases of \cref{thm:aligned_dephasing} simply distinguish whether it falls below or above $\sqrt2\, b$.

\emph{Case $(ii)$.}
If the root satisfies $\bar c \leq \sqrt2\, b$, solving $a + b - (\sqrt2 - 1) \bar c = 1$ gives
\begin{equation}
	\bar c = \frac{a + b - 1}{\sqrt2 - 1},
\end{equation}
and this assumption translates into the condition $a + (\sqrt2 - 1) b \leq 1$ quoted in the theorem.
The minimizer is $v = (1 + i)\, \bar c/\sqrt2$, which lies in $\mathcal V$ since $v_1 = \bar c /\sqrt 2 \leq b \leq a$, and the non-Clifford Kraus operator sits at
\begin{equation}
	z_\Phi = \frac{v}{|v|} = \frac{1 + i}{\sqrt 2}.
\end{equation}
It is a rotation of the Bloch sphere by $\pi/4$ about the $z$ axis, that is, a $T$ gate, irrespective of $\phi$ and $p$.
The Clifford component is the non-trivial mixture $z_\Sigma = (z_\Lambda - v)/(1 - \bar c)$, obtained from \cref{lemma:clifford_channel} through the substitution $a \to (a - \bar c/\sqrt2)/(1 - \bar c)$ and $b \to (b - \bar c/\sqrt2)/(1 - \bar c)$.

\emph{Case $(iii)$.}
Otherwise $\bar c \geq \sqrt2\, b$, and $a + \bar c - \sqrt{\bar c^2 - b^2} = 1$ gives
\begin{equation}
	\bar c = \frac{b^2 + (1-a)^2}{2(1-a)},
\end{equation}
which is well defined because $a < 1$.
The minimizer is $v = \sqrt{\bar c^2 - b^2} + i \, b$ with $\sqrt{\bar c^2 - b^2} = \bar c - (1 - a)$, and it lies in $\mathcal V$ as well: the only condition left to check, $v_1 \leq a$, amounts to $\bar c \leq 1$ and hence to $b^2 \leq 1 - a^2$, which is guaranteed by $|z_\Lambda| \leq 1$.
The non-Clifford Kraus operator now depends on both $\phi$ and $p$:
\begin{equation}
	z_\Phi = \frac{\sqrt{\bar c^2 - b^2} + i\, b}{\bar c},
\end{equation}
a rotation about the $z$ axis by an angle $2\theta$ with $\tan(2\theta) = b/\sqrt{\bar c^2 - b^2}$.
The Clifford component, instead, degenerates to the identity: the imaginary parts of $z_\Lambda$ and $v$ cancel and $a - v_1 = 1 - \bar c$, so that
\begin{equation}
	z_\Sigma = \frac{z_\Lambda - v}{1 - \bar c} = \frac{a - v_1}{1 - \bar c} = 1.
\end{equation}

This exhausts the three cases and proves \cref{thm:aligned_dephasing} in the sector $0 \leq b \leq a$, from which the symmetry argument above extends it to the whole complex plane.
The geometry of the two non-trivial regimes is illustrated in \cref{app_fig:aligned_minimizer}.
In \cref{app_fig:aligned_minimizer_T} the minimizer lies in the interior of $\mathcal V$, the optimal non-Clifford Kraus operator is a $T$ gate and the Clifford part is a non-trivial mixture of the four elements of $\mathcal G$.
In \cref{app_fig:aligned_minimizer_id} the constraint $v_2 = b$ is instead active, the Clifford part collapses onto the vertex $\mathds{1}$ and the non-Clifford Kraus operator is a generic rotation about $z$.

\begin{figure}[t!]
	\centering
	\includegraphics{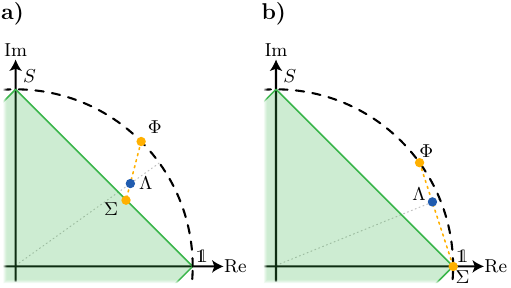}
	\caption{Optimal unraveling in the complex plane. The shaded region is the first quadrant of the Clifford diamond, the dashed line is the unit circle. The channel $\Lambda$ is decomposed into the non-Clifford component $\Phi$, of weight $\bar c$, and the Clifford mixture $\Sigma$. (a) Case $(ii)$ of \cref{thm:aligned_dephasing}: $\Phi$ is a $T$ gate and $\Sigma$ a non-trivial mixture of the four Clifford rotations. (b) Case $(iii)$: $\Sigma$ collapses onto the identity and $\Phi$ is a rotation about $z$ that depends on $\phi$ and $p$.}
	\label{app_fig:aligned_minimizer}
	{\phantomsubcaptionlabel{app_fig:aligned_minimizer_T}}
	{\phantomsubcaptionlabel{app_fig:aligned_minimizer_id}}
\end{figure}

\section{Proofs for Pauli noise isotropic on the equatorial plane}
\label{sec:proofs_pauli_noise}

This appendix proves the results stated in \cref{sec:pauli_noise}, following the same route as \cref{sec:proofs_aligned_dephasing}.
The noise is no longer confined to the $z$ axis, so \cref{lemma:z_rotations} does not apply and the Kraus operators of an unraveling can be arbitrary rotations.
What replaces it is a symmetry: the Bloch matrix $M_\Lambda$ associated to the channel $\Lambda$ commutes with every $g \in \mathcal G$, and averaging over $\mathcal G$ collapses $\so 3$ onto a two-parameter family at no cost, so that channels are again described by a complex number together with a height (\cref{app:pauli_symmetrization}).
In these coordinates the admissible channels fill the cone $\mathcal C$, the zero-cost ones the pyramid $\mathcal P \subset \mathcal C$, and a single non-Clifford Kraus operator again suffices (\cref{app:pauli_geometry}).
What is left is the planar problem of \cref{app_eq:aligned_optimization}, rescaled by the height of $\Lambda$, which we solve in \cref{app:pauli_optimum} to prove \cref{thm:pauli_noise}.

\subsection{Symmetrization and reduced coordinates}
\label{app:pauli_symmetrization}

Throughout the appendix $\sym(\cdot)$ denotes the average over conjugations by $\mathcal G$,
\begin{equation}
	\label{app_eq:symmetrization}
	\sym(M) = \frac{1}{4}\sum_{k=0}^3 R_z(k\pi/2)\, M \, R_z(-k\pi/2),
\end{equation}
a linear map, so that $\sym(\conv(\mathcal M)) = \conv(\sym(\mathcal M))$ for every $\mathcal M \subseteq \so3$.

\begin{lemma}[Symmetrization]
	\label{lemma:symmetrization}
	For every $m$, $\sym(\conv(\mathcal K_m)) \subseteq \conv(\mathcal K_m)$.
	Consequently, if $M_\Lambda = \sym(M_\Lambda)$, the value of \cref{eq:optimization_general} is unchanged when each $\conv(\mathcal K_m)$ is replaced by $\sym(\conv(\mathcal K_m))$.
\end{lemma}
\begin{proof}
	Conjugation by a Clifford operator maps the Bloch matrix associated to a Kraus operator $M_K = \tilde C_1 R_z(2\theta_1) \tilde C_2 R_z(2\theta_2) \tilde C_3 R_z(2\theta_3) \tilde C_4$ into one with the same number of rotations $R_z(2\theta)$, hence it preserves the cost.
	Since $\mathcal G$ belongs to the Clifford group, $R_z(k\pi/2) \, \mathcal K_m \, R_z(-k\pi/2) = \mathcal K_m$ for every $k$, and therefore $\sym(M) \in \conv(\mathcal K_m)$ for every $M \in \conv(\mathcal K_m)$.
	Let now $M_\Lambda = \sum_m p_m M_m$, with $M_m \in \conv(\mathcal K_m)$, be feasible for \cref{eq:optimization_general}.
	Applying $\sym$ to both sides and using $M_\Lambda = \sym(M_\Lambda)$ gives $M_\Lambda = \sum_m p_m \sym(M_m)$, with $\sym(M_m) \in \sym(\conv(\mathcal K_m))$, which is feasible for the restricted problem and carries the same cost $\sum_m p_m m$.
	The opposite inequality follows from $\sym(\conv(\mathcal K_m)) \subseteq \conv(\mathcal K_m)$, which makes every feasible point of the restricted problem feasible for the original one.
\end{proof}

The next lemma shows that symmetrized matrices are labelled by two coordinates only, and identifies the set they span.

\begin{lemma}[Reduced coordinates]
	\label{lemma:reduced_coordinates}
	For every $3\times3$ matrix $M$,
	\begin{equation}
		\label{app_eq:reduced_coordinates}
		\sym(M) = \begin{pmatrix}
			\Re(\zeta)  & \Im(\zeta) &   \\
			-\Im(\zeta) & \Re(\zeta) &   \\
			            &            & t
		\end{pmatrix},
	\end{equation}
	with $\zeta(M) = (M_{11}+M_{22})/2 + i (M_{12}-M_{21})/2$ and $t(M) = M_{33}$.
	Moreover, if $M \in \so3$, then
	\begin{equation}
		\label{app_eq:cone_identity}
		|\zeta(M)| = \frac{1+t(M)}{2},
	\end{equation}
	so that $\sym(\so3)$ is the lateral surface of the cone with apex $(0,-1)$ and base the unit circle at $t=1$, and
	\begin{equation}
		\label{app_eq:cone}
		\mathcal C := \conv(\sym(\so3)) = \qty{(\zeta, t) : |\zeta| \leq \frac{1+t}{2}, \, t \leq 1}.
	\end{equation}
\end{lemma}
\begin{proof}
	Split $M$ into its upper-left $2\times2$ block $B$, the vectors $u = (M_{13}, M_{23})^T$ and $v^T = (M_{31}, M_{32})$, and the scalar $t = M_{33}$, which conjugation by $R_z(k\pi/2)$ maps to
	\begin{flalign}
		 & \begin{aligned}
			   B   & \mapsto R_z(k\pi/2)\, B \, R_z(-k\pi/2), \quad & u & \mapsto R_z(k\pi/2)\, u, \\
			   v^T & \mapsto v^T R_z(-k\pi/2), \quad                & t & \mapsto t.
		   \end{aligned}
	\end{flalign}
	Since $\sum_k R_z(k\pi/2) = 0$, the averages of $u$ and $v$ vanish and only $B$ and $t$ survive.
	Decompose $B = H + A$, with
	\begin{equation}
		H = \Re(\zeta)\,\mathds{1} + \Im(\zeta)\, J, \quad J = \begin{pmatrix} 0 & 1\\ -1 & 0\end{pmatrix},
	\end{equation}
	and $A$ the traceless symmetric remainder; reading off the entries of $B$ gives $\Re(\zeta) = (B_{11}+B_{22})/2$ and $\Im(\zeta) = (B_{12}-B_{21})/2$, as in \cref{app_eq:reduced_coordinates}.
	Rotations commute with $H$, hence $\sym(H) = H$, whereas conjugation by $R_z(k\pi/2)$ rotates a traceless symmetric matrix by $k\pi$, that is, $A \mapsto (-1)^k A$, so that $\sym(A) = 0$.
	This proves \cref{app_eq:reduced_coordinates}.

	Let now $M \in \so3$ have rotation angle $\omega$ and axis $\hat n$.
	Its trace gives $2\Re(\zeta) + t = \Tr(M) = 1 + 2\cos(\omega)$, while Rodrigues' rotation formula gives $t = M_{33} = \cos(\omega) + (1-\cos(\omega))\,\hat n_3^2$ and, through $M - M^T = \pm 2 \sin(\omega) [\hat n]_\times$ with $[\hat n]_\times$ the generator of rotations about $\hat n$, $\Im(\zeta)^2 = \sin^2(\omega)\, \hat n_3^2$.
	The first two relations read
	\begin{equation}
		\begin{split}
			\frac{1+t}{2} - \Re(\zeta) & = (1 - \cos(\omega))\,\hat n_3^2, \\
			\frac{1+t}{2} + \Re(\zeta) & = 1 + \cos(\omega),
		\end{split}
	\end{equation}
	and multiplying them gives
	\begin{equation}
		\qty(\frac{1+t}{2})^2 - \Re(\zeta)^2 = (1 - \cos^2(\omega))\,\hat n_3^2 = \Im(\zeta)^2,
	\end{equation}
	both $(1+t)/2$ and $|\zeta|$ being non-negative, that is, \cref{app_eq:cone_identity}.
	Finally, at fixed $t$ the section of the right-hand side of \cref{app_eq:cone} is the disk of radius $(1+t)/2$, whose boundary circle belongs to $\sym(\so3)$: every point of that disk is a convex combination of two of its boundary points, and the set is convex, so it is the convex hull of the lateral surface.
\end{proof}

Any unraveling writes $M_\Lambda$ as a convex combination of rotations, so that $\Lambda$ itself is a point of the cone, $M_\Lambda \in \mathcal C$.
The Bloch matrix of \cref{eq:depolarizing} composed with $\mathcal U_\phi$ commutes with every $R_z(k\pi/2)$, hence $M_\Lambda = \sym(M_\Lambda)$ and \cref{lemma:symmetrization} applies; its coordinates read
\begin{equation}
	\label{app_eq:pauli_coordinates}
	(\zeta_\Lambda, t_\Lambda) = \qty(f_\bot e^{2i\phi},\, f_z).
\end{equation}
Its membership in $\mathcal C$, that is $|f_\bot| \leq 1 - 2p_\bot$, is guaranteed by $p_\bot, p_z \geq 0$ and $2p_\bot + p_z \leq 1$, and plays here the role that $|z_\Lambda| \leq 1$ played in \cref{app:aligned_optimum}.

\subsection{The Clifford pyramid}
\label{app:pauli_geometry}

The zero-cost points form the analogue of the Clifford diamond of \cref{lemma:clifford_channel}, sketched in \cref{fig:clifford_pyramid}.

\begin{lemma}[Clifford mixtures]
	\label{lemma:clifford_pyramid}
	The images of the $24$ Clifford Bloch matrices are the nine points
	\begin{equation}
		\label{app_eq:clifford_images}
		\begin{split}
			\sym(\mathcal K_0) = \qty{(i^k, 1)}_{k=0}^{3} & \cup \qty{(0,-1)}                        \\
			                                              & \cup \qty{(\pm 1/2, 0),\, (\pm i/2, 0)},
		\end{split}
	\end{equation}
	and their convex hull is the pyramid
	\begin{equation}
		\label{app_eq:pyramid}
		\mathcal P := \conv(\sym(\mathcal K_0)) = \qty{(\zeta, t) : ||\zeta||_1 \leq \frac{1+t}{2}, \, t \leq 1}.
	\end{equation}
	Moreover, introducing the height $s = (1+t)/2 \in [0,1]$, every $(\zeta, s) \in \mathcal P$ is the coordinate of the Clifford mixture
	\begin{equation}
		\label{app_eq:clifford_channel_pyramid}
		\Sigma = \alpha_0 I + \alpha_1 \mathcal S + \alpha_2 \mathcal Z + \alpha_3 \mathcal S^\dagger + (1-s)\, \mathcal A, \quad \mathcal A = (\mathcal X + \mathcal Y)/2,
	\end{equation}
	where $\alpha_k = s \, \beta_k$ and the $\beta_k$ are the weights that \cref{lemma:clifford_channel} assigns to the point $\zeta/s \in \mathcal D$.
\end{lemma}
\begin{proof}
	The elements of $\mathcal G$ commute with $\sym$ and sit at $(i^k, 1)$ by \cref{lemma:clifford_channel}; the Pauli $\sx$ and $\sy$, of Bloch matrices $\mathrm{diag}(1,-1,-1)$ and $\mathrm{diag}(-1,1,-1)$, both have $\zeta = 0$ and $t = -1$ and are mapped to the apex.
	Direct computation of the remaining Clifford Bloch matrices completes \cref{app_eq:clifford_images} with the four points $(\pm 1/2, 0)$ and $(\pm i/2, 0)$, which are the midpoints of the segments joining the apex to the base vertices and are therefore redundant in the convex hull.
	What is left are the four base vertices and the apex, whose hull is the square pyramid \cref{app_eq:pyramid}: its section at height $s$ is the $\ell_1$ ball of radius $s$.

	For the last statement, let $(\zeta, s) \in \mathcal P$ with $s > 0$, so that $||\zeta/s||_1 \leq 1$ and \cref{lemma:clifford_channel} provides non-negative weights $\beta_k$ summing to one with $(\beta_0-\beta_2) + i(\beta_3-\beta_1) = \zeta/s$.
	The weights $\alpha_k = s\beta_k$ and $1-s$ of \cref{app_eq:clifford_channel_pyramid} are then non-negative and sum to one.
	The four Clifford channels sit at $t=1$ and $\mathcal A$ at $t=-1$, so that the mixture has coordinates
	\begin{equation}
		\zeta_\Sigma = s\,\frac{\zeta}{s} = \zeta, \qquad t_\Sigma = s - (1-s) = 2s-1,
	\end{equation}
	as claimed; for $s=0$ the only point of $\mathcal P$ is the apex, $\Sigma = \mathcal A$.
\end{proof}

The cone is reached, in turn, by adding a single non-Clifford Kraus operator.

\begin{lemma}[Cost-one Kraus operators]
	\label{lemma:cost_one}
	$\conv(\sym(\mathcal K_0 \cup \mathcal K_1)) = \mathcal C$.
	Consequently, for every channel with $M_\Lambda = \sym(M_\Lambda)$ the optimum of \cref{eq:optimization_general} is attained with $p_m = 0$ for $m \geq 2$, that is
	\begin{equation}
		\label{app_eq:pauli_decomposition}
		\bar c = \min \qty{c : M_\Lambda = c\, \Phi + (1-c)\, \Sigma, \ \Phi \in \mathcal C, \ \Sigma \in \mathcal P},
	\end{equation}
	and in particular $\bar c \leq 1$.
\end{lemma}
\begin{proof}
	A rotation $R_z(\theta)$ costs at most one and $\sym(R_z(\theta)) = (e^{i\theta}, 1)$, which sweeps the whole base circle as $\theta$ varies, while the apex belongs to $\sym(\mathcal K_0)$ by \cref{lemma:clifford_pyramid}.
	The convex hull of the base circle and the apex is $\mathcal C$, and the reverse inclusion holds because $\sym(\mathcal K_m) \subseteq \mathcal C$ for every $m$ by \cref{lemma:reduced_coordinates}.
	Let now $M_\Lambda = \sum_m p_m M_m$ be feasible for the symmetrized problem of \cref{lemma:symmetrization}, so that $M_m \in \sym(\conv(\mathcal K_m)) \subseteq \mathcal C$.
	Since $\mathcal C = \conv(\sym(\mathcal K_0 \cup \mathcal K_1))$, each $M_m$ splits as $M_m = \mu_m \Phi_m + (1-\mu_m)\Sigma_m$ with $\Phi_m \in \mathcal C$, $\Sigma_m \in \mathcal P$ and $\mu_m \in [0,1]$, and in particular $\mu_0 = 0$ because $M_0 \in \mathcal P$.
	Collecting the two groups produces a decomposition of the form \cref{app_eq:pauli_decomposition} of cost $\sum_m p_m \mu_m \leq \sum_m p_m m$, since $\mu_m \leq 1 \leq m$ for $m \geq 1$.
	Finally, $M_\Lambda \in \mathcal C$ makes $c=1$, $\Phi = M_\Lambda$ admissible, whence $\bar c \leq 1$.
\end{proof}

\subsection{Proof of the optimal unraveling and cost}
\label{app:pauli_optimum}

In the height coordinate $s = (1+t)/2$ the two bodies of \cref{app_eq:pauli_decomposition} read
\begin{equation}
	\label{app_eq:cone_pyramid_height}
	\begin{split}
		\mathcal C & = \qty{(\zeta, s) : |\zeta| \leq s \leq 1},     \\
		\mathcal P & = \qty{(\zeta, s) : ||\zeta||_1 \leq s \leq 1},
	\end{split}
\end{equation}
so that $\mathcal C$ and $\mathcal P$ are the cone over the unit disk $\mathcal B$ and over the Clifford diamond $\mathcal D$ of \cref{app_eq:diamond}, respectively.
As in \cref{app:aligned_reduction} we absorb the weight of the non-Clifford Kraus operators into its coordinates, $v = \bar c \, \zeta_\Phi$ and $\eta = \bar c \, s_\Phi$, so that $\Phi \in \mathcal C$ and $\Sigma \in \mathcal P$ become
\begin{equation}
	\label{app_eq:pauli_constraints}
	|v| \leq \eta \leq \bar c, \qquad ||\zeta_\Lambda - v||_1 \leq s_\Lambda - \eta \leq 1 - \bar c.
\end{equation}
Lowering $\eta$ relaxes the second constraint and is only limited by $\eta \geq |v|$; at $\eta = |v|$ the first constraint reads $\bar c \geq |v|$ and the rightmost one, $s_\Lambda - |v| \leq 1 - \bar c$, is satisfied at $\bar c = |v|$ because $s_\Lambda \leq 1$.
The optimum therefore saturates $\bar c = \eta = |v|$ and
\begin{equation}
	\label{app_eq:pauli_optimization}
	\bar c = \min_{v \in \mathbb{C}} |v| \quad \text{subject to} \quad ||\zeta_\Lambda - v||_1 + |v| \leq s_\Lambda,
\end{equation}
a minimizer of which reconstructs the unraveling through
\begin{equation}
	\label{app_eq:pauli_reconstruction}
	(\zeta_\Phi, s_\Phi) = \qty(\frac{v}{|v|}, 1), \qquad
	(\zeta_\Sigma, s_\Sigma) = \qty(\frac{\zeta_\Lambda - v}{1 - \bar c}, \frac{s_\Lambda - \bar c}{1 - \bar c}).
\end{equation}
The non-Clifford Kraus operator lies on the base circle of the cone: it is a rotation about the $z$ axis, exactly as in \cref{app:aligned_optimum}, and the Clifford part is recovered from \cref{lemma:clifford_pyramid}.

\cref{app_eq:pauli_optimization} is \cref{app_eq:aligned_optimization} with the unit diamond replaced by the diamond of size $s_\Lambda$.
Assuming $s_\Lambda > 0$, since $s_\Lambda = 0$ forces $\zeta_\Lambda = 0$ and $\bar c = 0$, the substitution $\tilde \zeta_\Lambda = \zeta_\Lambda/s_\Lambda$ and $\tilde v = v / s_\Lambda$ maps one problem into the other,
\begin{equation}
	\label{app_eq:pauli_rescaled}
	\bar c = s_\Lambda \, \min_{\tilde v \in \mathbb{C}} \qty{|\tilde v| : ||\tilde \zeta_\Lambda - \tilde v||_1 + |\tilde v| \leq 1},
\end{equation}
and \cref{app_eq:cone_identity} guarantees $|\tilde \zeta_\Lambda| \leq 1$, the hypothesis used in \cref{app:aligned_optimum}.
The solution obtained there applies verbatim to $\tilde \zeta_\Lambda = (A + i B)/s_\Lambda$, with $A = |f_\bot \cos(2\phi)|$ and $B = |f_\bot \sin(2\phi)|$ and $0 \leq B \leq A$ as in \cref{thm:pauli_noise}.
Substituting $(a, b) = (A, B)/s_\Lambda$ in \cref{thm:aligned_dephasing} and multiplying the cost by $s_\Lambda$ turns its three conditions and costs into
\begin{equation}
	\begin{split}
		a + b \leq 1             & \ \longrightarrow \ A + B \leq s_\Lambda,              \\
		a + (\sqrt2 -1) b \leq 1 & \ \longrightarrow \ A + (\sqrt2 - 1) B \leq s_\Lambda,
	\end{split}
\end{equation}
and
\begin{equation}
	\begin{split}
		s_\Lambda\,\frac{a+b-1}{\sqrt2 - 1}     & \longrightarrow \frac{A + B - s_\Lambda}{\sqrt 2 - 1},            \\
		s_\Lambda\,\frac{b^2 + (1-a)^2}{2(1-a)} & \longrightarrow \frac{B^2 + (s_\Lambda - A)^2}{2(s_\Lambda - A)},
	\end{split}
\end{equation}
which is \cref{eq:cost_pauli_noise}.
It remains to read off the unraveling in each case from \cref{app_eq:pauli_reconstruction}.

\emph{Case $(i)$.}
Here $v = 0$ and $\bar c = 0$, so that $\Sigma = \Lambda$ and \cref{lemma:clifford_pyramid} returns \cref{eq:clifford_unraveling_pyramid}: the weights $\beta_k$ of \cref{lemma:clifford_channel} for $\tilde \zeta_\Lambda$ carry the slack $1 - a - b$, so that $\alpha_k = s_\Lambda \beta_k$ carries $\tau = s_\Lambda - A - B$.

\emph{Case $(ii)$.}
The minimizer is $\tilde v = (1+i)\,\bar c/(\sqrt2 \, s_\Lambda)$, so that $\zeta_\Phi = (1+i)/\sqrt2$ and $s_\Phi = 1$: the non-Clifford Kraus operator is again a $T$ gate, irrespective of $\phi$, $p_\bot$ and $p_z$.
The Clifford part is the mixture \cref{app_eq:clifford_channel_pyramid} of coordinates \cref{app_eq:pauli_reconstruction}, that is, with the substitutions
\begin{equation}
	\label{app_eq:pauli_case_ii_substitution}
	A \to \frac{A - \bar c/\sqrt2}{1 - \bar c}, \qquad
	B \to \frac{B - \bar c/\sqrt2}{1 - \bar c}, \qquad
	s_\Lambda \to \frac{s_\Lambda - \bar c}{1 - \bar c}.
\end{equation}

\emph{Case $(iii)$.}
The minimizer now satisfies $\zeta_\Lambda - v = s_\Lambda - \bar c$, which is real, so that \cref{app_eq:pauli_reconstruction} gives $\zeta_\Sigma = s_\Sigma$: the Clifford part sits on the edge of $\mathcal P$ joining the apex to the vertex $\mathds{1}$,
\begin{equation}
	\label{app_eq:pauli_case_iii_sigma}
	\Sigma = \frac{s_\Lambda - \bar c}{1 - \bar c}\, I + \frac{1 - s_\Lambda}{1 - \bar c}\, \mathcal A,
\end{equation}
while the non-Clifford Kraus operator is the rotation about $z$ by the angle $2\theta$ with $\tan(2\theta) = B/\sqrt{\bar c^2 - B^2}$ found in \cref{app:aligned_optimum}.

This proves \cref{thm:pauli_noise}.
Its close resemblance to \cref{thm:aligned_dephasing} is now transparent: the whole effect of the noise components that do not preserve the $z$ axis is to lower the height $s_\Lambda$ of the channel, thereby shrinking the diamond that the Clifford part has at its disposal, without changing the geometry of the problem.
Aligned dephasing is recovered at $p_\bot = 0$, where $s_\Lambda = 1$ and $\mathcal P$ is cut at its base.

\section{Proofs for tilted dephasing}
\label{sec:proofs_tilted_dephasing}

This appendix proves the results stated in \cref{sec:tilted_dephasing}.
The noise axis $\vn$ is now generic, so that $M_\Lambda$ does not commute with $\mathcal G$ and neither the picture of the Clifford diamond of \cref{sec:proofs_aligned_dephasing} nor the cone of \cref{sec:proofs_pauli_noise} are available.
Yet, even in this case there is a relation that ties the structure of the channel to the geometry of the problem: $\Lambda$ maps a specific pure state onto another pure state, and this alone pins every Kraus operator of every unraveling to a one-parameter family (\cref{lemma:tilted_rigidity}).
An unraveling is then a probability measure on a circle, and the cost of each Kraus operator is fixed by an algebraic property of its Bloch matrix (\cref{lemma:cost_criterion}).

\subsection{Reduction to a measure on the circle}
\label{app:tilted_reduction}

In the Bloch matrix representation the composite channel $\Lambda = \mathcal N_{p, \vn} \circ \mathcal U_\phi$ reads
\begin{equation}
	\label{app_eq:tilted_bloch}
	M_\Lambda = N_{p,\vn}\, R_z(2\phi), \qquad N_{p,\vn} = (1-2p)\,\mathds{1} + 2p\, \vn \vn^T,
\end{equation}
so that the noise leaves the axis $\vn$ invariant and shrinks by $1-2p$ every unit vector $\vn_\bot$ orthogonal to it,
\begin{equation}
	\label{app_eq:tilted_noise_action}
	N_{p,\vn}\, \vn = \vn, \qquad N_{p,\vn}\, \vn_\bot = (1-2p)\, \vn_\bot.
\end{equation}
In particular $M_\Lambda\, v = \vn$ for $v = R_z(-2\phi)\vn$: the channel maps the pure states $\pm v$ onto the pure states $\pm \vn$ without shrinking them, and it is this single rigid direction that constrains all of its unravelings.

\begin{lemma}[Rigidity of the Kraus operators]
	\label{lemma:tilted_rigidity}
	Let $\Lambda(\rho) = \sum_j \lambda_j K_j \rho K_j^\dagger$ be any unraveling of $\Lambda$ into unitary Kraus operators and let $R_j \in \so3$ be the Bloch matrix of $K_j$.
	Then $R_j\, v = \vn$ for every $j$ and, consequently,
	\begin{equation}
		\label{app_eq:tilted_circle}
		R_j \in \mathcal O := \qty{M_\psi = R_{\vn}(\psi)\, R_z(2\phi) : \psi \in [0, 2\pi)},
	\end{equation}
	the corresponding unitaries being $V(\psi) = e^{i \psi\, \vn \cdot \vsigma/2}\, U_\phi$.
\end{lemma}
\begin{proof}
	Let $u_j = R_j v$, so that $|u_j| = 1$ and, by linearity, $\sum_j \lambda_j u_j = M_\Lambda v = \vn$.
	Taking the norm,
	\begin{equation}
		1 = \Big|\sum_j \lambda_j u_j\Big| \leq \sum_j \lambda_j |u_j| = \sum_j \lambda_j = 1,
	\end{equation}
	so that the triangle inequality is saturated: all the $u_j$ are parallel and equally oriented, and $u_j = \vn$ for every $j$.
	Since $R_z(2\phi)\, v = \vn$ as well, the rotation $S_j = R_j R_z(-2\phi)$ satisfies $S_j \vn = \vn$ and is therefore a rotation about $\vn$ by some angle $\psi_j$, that is, $R_j = R_{\vn}(\psi_j) R_z(2\phi)$.
\end{proof}

Two Kraus operators deserve a name: $\psi = 0$ gives $V(0) = U_\phi$, while $\psi = \pi$ gives $V(\pi) \propto (\vn\cdot\vsigma)\, U_\phi$, the Kraus operator of $\mathcal N_{p, \vn}$ applied after the gate.
The naive unraveling of \cref{sec:aligned_dephasing} is thus the measure $\mu = (1-p)\,\delta_0 + p\,\delta_\pi$.
By \cref{lemma:tilted_rigidity} a generic unraveling is nothing but a probability measure $\mu$ on the circle, and the next lemma determines the admissible ones.

\begin{lemma}[Admissible unravelings]
	\label{lemma:tilted_measure}
	A probability measure $\mu$ on $[0, 2\pi)$ defines an unraveling of $\Lambda$ if and only if
	\begin{equation}
		\label{app_eq:tilted_constraints}
		\int \dd\mu(\psi) = 1, \qquad \int \dd\mu(\psi)\, e^{i\psi} = 1-2p.
	\end{equation}
\end{lemma}
\begin{proof}
	Averaging the Kraus operators produces the channel of Bloch matrix $\int \dd\mu(\psi)\, M_\psi$, which by \cref{app_eq:tilted_bloch,app_eq:tilted_circle} equals $M_\Lambda$ if and only if $\int \dd\mu(\psi)\, R_{\vn}(\psi) = N_{p,\vn}$.
	In the orthonormal basis $\qty{\vn, \vn_\bot, \vn\times\vn_\bot}$, in which \cref{app_eq:tilted_noise_action} is diagonal, the two matrices read
	\begin{equation}
		\begin{split}
			R_{\vn}(\psi) & = \begin{pmatrix}
				                  1 &  & \\ & \cos\psi &\sin\psi \\ & - \sin\psi & \cos\psi
			                  \end{pmatrix}, \\
			N_{p,\vn}     & = \begin{pmatrix}
				                  1 &  & \\ & 1-2p & \\ & & 1-2p
			                  \end{pmatrix},
		\end{split}
	\end{equation}
	and equating them entry by entry gives \cref{app_eq:tilted_constraints}.
\end{proof}

Splitting the second constraint into its real and imaginary parts turns \cref{app_eq:tilted_constraints} into the two moment conditions of \cref{eq:optimization_tilted},
\begin{equation}
	\int \dd\mu(\psi)\, \qty(1 - \cos\psi) = 2p, \qquad \int \dd\mu(\psi)\, \sin\psi = 0.
\end{equation}

\subsection{The cost of a Kraus operator}
\label{app:tilted_cost}

In \cref{sec:proofs_aligned_dephasing,sec:proofs_pauli_noise} every Kraus operator was either Clifford or a single rotation about $z$, so that the cost of an unraveling was just the total weight of its non-Clifford Kraus operators.
Here $\mathcal O$ meets all the cost classes $\mathcal K_m$ and the cost of $V(\psi)$ is a non-trivial function of $\psi$, which the following lemma reduces to the inspection of the entries of $M_\psi$.

\begin{lemma}[Cost criterion]
	\label{lemma:cost_criterion}
	Let $c(M)$ be the cost of $M \in \so3$ introduced in \cref{sec:pauli_noise}, that is, the minimal number of rotations $R_z(2\theta_j)$ in a decomposition $M = \tilde C_1 R_z(2\theta_1) \tilde C_2 R_z(2\theta_2) \tilde C_3 R_z(2\theta_3) \tilde C_4$ into Clifford Bloch matrices $\tilde C_j$.
	Then
	\begin{equation}
		\label{app_eq:cost_criterion}
		c(M) = \begin{dcases}
			0 & \text{if } M \text{ is Clifford},             \\
			1 & \text{else if } \exists\, i,j : |M_{ij}| = 1, \\
			2 & \text{else if } \exists\, i,j : M_{ij} = 0,   \\
			3 & \text{otherwise}.
		\end{dcases}
	\end{equation}
\end{lemma}
\begin{proof}
	Since a product of Clifford operators is Clifford, $c(M) = 0$ if and only if $M$ is Clifford.
	We prove that $c(M) \leq 1 \iff \exists\, i,j : |M_{ij}| = 1$, that $c(M) \leq 2 \iff \exists\, i,j : M_{ij} = 0$, and that $c(M) \leq 3$ always, which together give \cref{app_eq:cost_criterion}.

	Let $M = \tilde C_1 R_z(2\theta) \tilde C_2$.
	Since $R_z(2\theta) e_3 = e_3$, we have $M \tilde C_2^{-1} e_3 = \tilde C_1 e_3$, and both $\tilde C_2^{-1} e_3$ and $\tilde C_1 e_3$ are signed coordinate vectors, say $\pm e_j$ and $\pm e_i$: hence $|M_{ij}| = 1$.
	Conversely, $|M_{ij}| = 1$ forces the remaining entries of that row and column to vanish by orthogonality, so that $M e_j = \pm e_i$.
	Choosing Clifford matrices with $\tilde C_1 e_3 = \pm e_i$ and $\tilde C_2^{-1} e_3 = e_j$, the rotation $G = \tilde C_1^{-1} M \tilde C_2^{-1}$ fixes $e_3$ and is therefore a rotation about it, $G = R_z(2\theta)$, whence $M = \tilde C_1 R_z(2\theta) \tilde C_2$.

	Let now $M = \tilde C_1 R_z(2\theta_1) \tilde C_2 R_z(2\theta_2) \tilde C_3$ and set $G = \tilde C_1^{-1} M \tilde C_3^{-1} = R_z(2\theta_1) \tilde C_2 R_z(2\theta_2)$.
	Rotations about $e_3$ leave it invariant, so that $e_3^T G e_3 = (\tilde C_2)_{33} \in \{0, \pm 1\}$; reading this back on $M$ gives a pair with $|M_{ij}| \in \{0, 1\}$.
	Either $M_{ij} = 0$, in which case the statement is proven, or $|M_{ij}| = 1$ and consequently all other entries in the same row and column vanish.
	Conversely, let $M_{ij} = 0$ and pick Clifford matrices with $\tilde C_1 e_3 = e_i$ and $\tilde C_3^{-1} e_3 = e_j$, so that $G = \tilde C_1^{-1} M \tilde C_3^{-1}$ satisfies $e_3^T G e_3 = 0$, i.e. $G e_3$ lies in the equatorial plane.
	Choosing a Clifford $\tilde C_2$ with $\tilde C_2 e_3$ equatorial as well, and $\theta_1$ such that $R_z(2\theta_1) \tilde C_2 e_3 = G e_3$, the rotation $\qty(R_z(2\theta_1) \tilde C_2)^{-1} G$ fixes $e_3$ and equals some $R_z(2\theta_2)$, whence $M = \tilde C_1 R_z(2\theta_1) \tilde C_2 R_z(2\theta_2) \tilde C_3$.

	Finally, every $M \in \so3$ admits the Euler decomposition $M = R_z(2\theta_1) R_x(2\theta_2) R_z(2\theta_3)$, and $R_x(2\theta) = \tilde C R_z(2\theta) \tilde C^{-1}$ for the Clifford $\tilde C$ exchanging the $x$ and $z$ axes, so that $c(M) \leq 3$.
\end{proof}

Writing $c(\psi) = c(M_\psi)$, \cref{lemma:tilted_measure,lemma:cost_criterion} turn the search for the optimal unraveling into \cref{eq:optimization_tilted},
\begin{equation}
	\label{app_eq:tilted_optimization}
	\bar c = \min_{\mu} \int_0^{2\pi} \dd\mu(\psi)\, c(\psi),
\end{equation}
the minimum running over the probability measures obeying \cref{app_eq:tilted_constraints}.
We have no closed form solution for it.
Numerically, \cref{app_eq:cost_criterion} allows one to enumerate the angles with $c(\psi) = 0, 1, 2$; the remaining ones, of cost three, are sampled uniformly on the circle, and \cref{app_eq:tilted_optimization} is solved over the resulting finite set of Kraus operators.

The criterion has an immediate consequence for the equatorial noise axes, $\theta = \pi/2$, singled out in \cref{sec:tilted_dephasing}.
There $M_\pi = R_{\vn}(\pi) R_z(2\phi)$ is a rotation by $\pi$ about the equatorial axis of azimuth $\varphi + \phi$, so that $(M_\pi)_{33} = -1$ and $c(M_\pi) \leq 1$ for every $\varphi$, whereas a generic axis gives $c(M_\pi) = 3$.
The naive unraveling then costs
\begin{equation}
	\bar c \leq (1-p)\, c(M_0) + p\, c(M_\pi) \leq 1,
\end{equation}
which is why any equatorial axis can be disentangled at least up to $N_T \sim N$.
A rotation by $\pi$ about an equatorial axis is moreover Clifford if and only if its azimuth is an integer multiple of $\pi/4$, so that $c(M_\pi) = 0$ and $\bar c \leq 1-p$ when
\begin{equation}
	\varphi + \phi \in \frac{\pi}{4}\mathbb{Z},
\end{equation}
that is, at $\varphi = \pm\pi/8, \pm3\pi/8$ for $\phi = -\pi/8$: there the non-Clifford $\pi$ rotation $\vn\cdot\vsigma$ and the non-Clifford gate $U_\phi$ compose into a single Clifford operation.

\subsection{Clifford twirling}
\label{app:twirling}

It remains to prove the two statements of \cref{sec:twirled_dephasing}: twirling can only lower the cost, and $\mathcal G$-twirling maps $\mathcal N_{p,\vn}$ onto a Pauli channel isotropic on the equator, for which \cref{thm:pauli_noise} then holds verbatim.
Throughout, $g^\dagger \Lambda g$ denotes the channel $\rho \mapsto g^\dagger \Lambda(g \rho g^\dagger) g$, so that \cref{eq:twirling} reads $\Lambda^{\rm tw} = |G|^{-1}\sum_{g \in G} g^\dagger \Lambda g$.

\begin{lemma}[Twirling]
	\label{lemma:twirling}
	The optimal cost is convex and invariant under Clifford conjugation,
	\begin{equation}
		\label{app_eq:cost_convexity}
		\bar c\qty(\sum_i \mu_i \Lambda_i) \leq \sum_i \mu_i\, \bar c(\Lambda_i), \qquad \bar c\qty(g^\dagger \Lambda g) = \bar c(\Lambda),
	\end{equation}
	for $\mu_i \geq 0$, $\sum_i \mu_i = 1$ and $g$ Clifford.
	Consequently $\bar c(\Lambda^{\rm tw}) \leq \bar c(\Lambda)$ for any twirling set $G$ of Clifford operators.
\end{lemma}
\begin{proof}
	Let $\Lambda_i(\rho) = \sum_j \lambda_{ij} K_{ij}\rho K_{ij}^\dagger$ be optimal unravelings.
	The Kraus operators $K_{ij}$ with weights $\mu_i \lambda_{ij}$ unravel $\sum_i \mu_i \Lambda_i$ at cost $\sum_{ij}\mu_i\lambda_{ij}\, c(M_{ij}) = \sum_i \mu_i \bar c(\Lambda_i)$, and the minimum over all unravelings can only be smaller.
	Conjugation by a Clifford maps the unravelings of $\Lambda$ bijectively onto those of $g^\dagger \Lambda g$, sending each Kraus operator to a Clifford conjugate of the same cost, cf. the proof of \cref{lemma:symmetrization}: hence the second identity.
	Since $\Lambda^{\rm tw}$ is a convex combination of the channels $g^\dagger \Lambda g$, the two properties give $\bar c(\Lambda^{\rm tw}) \leq \sum_g \bar c(g^\dagger \Lambda g)/|G| = \bar c(\Lambda)$.
\end{proof}

Two requirements fix the twirling set.
Its elements must be Clifford, so that twirling injects no non-Clifford operation of its own, and they must leave the gate untouched, $g^\dagger U_\phi g = U_\phi$, so that twirling acts on the noise alone,
\begin{equation}
	g^\dagger \Lambda g = \qty(g^\dagger \mathcal N_{p,\vn}\, g) \circ \mathcal U_\phi.
\end{equation}
For generic $\phi$ the second requirement forces $g$ to preserve the $z$ axis, and \cref{lemma:z_rotations} identifies the Clifford operators that do so with the four elements of $\mathcal G$: this is the twirling set used in \cref{sec:twirled_dephasing}.

\begin{lemma}[Twirled noise]
	\label{lemma:twirled_noise}
	$\mathcal G$-twirling maps the tilted dephasing channel onto
	\begin{equation}
		\label{app_eq:twirled_channel}
		\mathcal N_{p,\vn}^{\rm tw} = (1-p)\, I + p\,\frac{\sin^2\theta}{2}\qty(\mathcal X + \mathcal Y) + p \cos^2\theta\, \mathcal Z,
	\end{equation}
	that is, onto the channel of \cref{eq:depolarizing} with
	\begin{equation}
		\label{app_eq:twirled_probabilities}
		p_\bot = \frac{p}{2}\sin^2\theta, \qquad p_z = p\cos^2\theta,
	\end{equation}
	whose PTM is $\mathrm{diag}(1, f_\bot, f_\bot, f_z)$ with $f_\bot = 1 - p(1+\cos^2\theta)$ and $f_z = 1 - 2p\sin^2\theta$.
\end{lemma}
\begin{proof}
	Conjugating the noise operator $\vn \cdot \vsigma$ by the elements of $\mathcal G$ gives
	\begin{equation}
		\begin{split}
			\mathds{1}\, (\vn\cdot\vsigma)\, \mathds{1} & = n_x \sx + n_y \sy + n_z \sz,  \\
			S^\dagger (\vn\cdot\vsigma)\, S             & = n_y \sx - n_x \sy + n_z \sz,  \\
			\sz\, (\vn\cdot\vsigma)\, \sz               & = -n_x \sx - n_y \sy + n_z \sz, \\
			S\, (\vn\cdot\vsigma)\, S^\dagger           & = -n_y \sx + n_x \sy + n_z \sz,
		\end{split}
	\end{equation}
	four operators $N_k = a_k \sx + b_k \sy + n_z \sz$ in which $(a_k, b_k)$ runs over the four rotations of $(n_x, n_y)$ by multiples of $\pi/2$.
	Twirling acts on the noise alone, so that
	\begin{equation}
		\mathcal N_{p,\vn}^{\rm tw}(\rho) = (1-p)\,\rho + \frac{p}{4}\sum_{k=0}^3 N_k \rho N_k,
	\end{equation}
	and in the expansion of $N_k \rho N_k$ over the Pauli basis the mixed terms cancel,
	\begin{equation}
		\sum_k a_k b_k = \sum_k a_k n_z = \sum_k b_k n_z = 0,
	\end{equation}
	while the diagonal ones give
	\begin{equation}
		\frac{1}{4}\sum_k a_k^2 = \frac{1}{4}\sum_k b_k^2 = \frac{n_x^2 + n_y^2}{2} = \frac{\sin^2\theta}{2}, \quad n_z^2 = \cos^2\theta,
	\end{equation}
	which is \cref{app_eq:twirled_channel}.
	Reading off $p_\bot$ and $p_z$ and inserting them in $f_\bot = 1 - 2(p_\bot + p_z)$ and $f_z = 1 - 4 p_\bot$ gives the PTM.
\end{proof}

The azimuth $\varphi$ has disappeared, as it must for a channel isotropic on the equator, and every result of \cref{sec:proofs_pauli_noise} applies to $\Lambda^{\rm tw}$ with the substitution \cref{app_eq:twirled_probabilities}.
Together with \cref{lemma:twirling}, this is what \cref{fig:twirled_dephasing} compares: the twirled cost is never larger than the original one, and it vanishes on an extended range of $p$ whenever \cref{app_eq:twirled_probabilities} places the channel inside the Clifford pyramid.

\section{Proofs for nonstabilizerness measures}
\label{sec:proofs_magic}

This appendix proves the results stated in \cref{sec:magic}: the closed forms of \cref{eqs:robustness_closed} and \cref{thm:magic}.
The route is the one of \cref{sec:proofs_pauli_noise}: for the channels of \cref{eq:depolarizing} the RoM of the Choi state is again a function of the reduced coordinates $(\zeta_\Lambda, s_\Lambda)$ of \cref{lemma:reduced_coordinates}, and it measures the amount by which the channel sticks out of the Clifford pyramid $\mathcal P$ (\cref{app:rom_closed}).
Since the same excess controls the cost of \cref{thm:pauli_noise}, the two quantifiers are locked together (\cref{app:rom_cost}).

\subsection{The robustness of magic in reduced coordinates}
\label{app:rom_closed}

Define the maximally entangled state $\ketbra{\Phi}{\Phi} = (\mathds{1}\otimes\mathds{1} + \sx\otimes\sx - \sy\otimes\sy + \sz\otimes\sz)/4$.
The Choi state of the unital channel $\Lambda$ is $J_\Lambda = (\Lambda\otimes\mathds{1})\ketbra{\Phi}{\Phi}$.
Since $\Lambda(\sigma_j) = \sum_i (M_\Lambda)_{ij}\, \sigma_i$, we find
\begin{equation}
	\label{app_eq:choi_bloch}
	J_\Lambda = \frac{1}{4}\Big(\mathds{1}\otimes\mathds{1} + \sum_{i,j=1}^3 (M_\Lambda)_{ij}\, \eta_j\, \sigma_i \otimes \sigma_j\Big),
\end{equation}
with $\eta = (1,-1,1)$.
\cref{app_eq:choi_bloch} is linear in $M_\Lambda$, and Bloch matrices combine linearly as well, so that any decomposition of the channel is inherited by its Choi state: if $\Lambda = \sum_j \lambda_j \mathcal C_j$, then $J_\Lambda = \sum_j \lambda_j J_{\mathcal C_j}$.
Notice that everywhere in \cref{sec:classical_phases} we wrote the decomposition of a generic channel $\Lambda = \bar c \Phi + (1 - \bar c) \Sigma$, with $\Phi$ and $\Sigma$ a generic non-Clifford unitary, and a Clifford mixture respectively.
Here, we use a different decomposition: any unital channel can be decomposed as a mixture of Clifford unitaries with possibly negative weights $\Lambda = \sum_j \lambda_j \mathcal C_j$, with $\lambda_j \in \mathbb{R}$ and $\sum_j \lambda_j = 1$~\cite{howard2017application, bravyi2016trading}.
The amount of negativity that is present in the decomposition measures precisely how non-Clifford is the channel $\Lambda$.

\begin{lemma}[Signed Clifford decompositions]
	\label{lemma:signed_clifford}
	Let $\Lambda = \sum_j \lambda_j \, \mathcal C_j$, where $\mathcal C_j(\rho) = C_j \rho C_j^\dagger$ are Clifford channels, $\lambda_j \in \mathbb{R}$ and $\sum_j \lambda_j = 1$, and call $\nu = \sum_j |\lambda_j|$ the weight of the decomposition.
	Then $\mathcal R(J_\Lambda) \leq \nu$.
	Moreover, if $\Lambda$ and $\Lambda'$ admit decompositions of weights $\nu$ and $\nu'$, then $\Lambda \circ \Lambda'$ admits one of weight $\nu \nu'$.
\end{lemma}
\begin{proof}
	The Choi state $J_{\mathcal C_j} = (C_j \otimes \mathds{1}) \ketbra{\Phi}{\Phi} (C_j \otimes \mathds{1})^\dagger$ is a pure stabilizer state, so that $J_\Lambda = \sum_j \lambda_j J_{\mathcal C_j}$ is admissible in \cref{eq:robustness_of_magic} and $\mathcal R(J_\Lambda) \leq \nu$.
	For the second statement, $\Lambda \circ \Lambda' = \sum_{jl} \lambda_j \lambda'_l \, \mathcal C_j \circ \mathcal C'_l$ is again a combination of Clifford channels, whose coefficients sum to one and whose absolute values sum to $\nu \nu'$.
\end{proof}

Note that $\mathcal R \geq 1$ always, since $\sum_k |x_k| \geq |\sum_k x_k| = 1$.
The signed extension of \cref{lemma:clifford_channel} is immediate: for any $z = a+ib$ in the unit disk its weights $\beta_k$, now carrying the slack $\tau = 1 - ||z||_1$ of either sign, still reproduce $z$ and still sum to one.
Moreover, since the term that carries $|a|$ equals $(1 + |a| - |b|)/2$, non-negative because $|b|\leq 1$, so at most one weight turns negative and the slack cancels in $\sum_j |\beta_j|$.
Therefore, we find
\begin{equation}
	\label{app_eq:signed_diamond}
	\sum_k |\beta_k| = \max\qty(||z||_1, \, 1).
\end{equation}

The matching lower bound comes from a witness.

\begin{lemma}[Witness]
	\label{lemma:rom_witness}
	Let
	\begin{equation}
		\label{app_eq:rom_witness}
		\begin{split}
			W = & \frac{1}{2}\qty(\sx\otimes\sx - \sy\otimes\sy - \sx\otimes\sy - \sy\otimes\sx) \\
			    & + \frac{1}{2}\qty(\mathds{1}\otimes\mathds{1} - \sz\otimes\sz).
		\end{split}
	\end{equation}
	Then $|\Tr(W \sigma)| \leq 1$ for every pure two-qubit stabilizer state $\sigma$, and consequently $\mathcal R(\rho) \geq \Tr(W\rho)$ for every two-qubit state $\rho$.
\end{lemma}
\begin{proof}
	The second statement follows from the first: if $\rho = \sum_k x_k \sigma_k$ is admissible in \cref{eq:robustness_of_magic}, then $\Tr(W\rho) \leq \sum_k |x_k| \, |\Tr(W \sigma_k)| \leq \sum_k |x_k|$, regardless of the choice of decomposition.
	Thus, the relation will hold in particular for the decomposition that minimizes $\sum_k |x_k|$, i.e., $\mathcal R(\rho)$.
	For the first, let $S$ be the stabilizer group of $\sigma$, so that $\ev{P} := \Tr(P \sigma) = \pm 1$ if $\pm P \in S$ and vanishes otherwise, and write
	\begin{equation}
		\Tr(W\sigma) = \frac{a - b - d - e}{2} + \frac{1 - c}{2},
	\end{equation}
	with $a = \ev{\sx\otimes\sx}$, $b = \ev{\sy\otimes\sy}$, $d = \ev{\sx\otimes\sy}$, $e = \ev{\sy\otimes\sx}$ and $c = \ev{\sz\otimes\sz}$.
	The five Pauli operators involved obey
	\begin{equation}
		(\sx\otimes\sx)(\sy\otimes\sy) = -\sz\otimes\sz, \qquad (\sx\otimes\sy)(\sy\otimes\sx) = \sz\otimes\sz,
	\end{equation}
	while $\sx\otimes\sx$ and $\sy\otimes\sy$ both anticommute with $\sx\otimes\sy$ and $\sy\otimes\sx$.
	Since $S$ is abelian, either $a = b = 0$ or $d = e = 0$.
	In each of the two triples $\{\sx\otimes\sx, \sy\otimes\sy, \sz\otimes\sz\}$ and $\{\sx\otimes\sy, \sy\otimes\sx, \sz\otimes\sz\}$ the product of any two elements is proportional to the third, so that the surviving triple contributes $0$, $1$ or $3$ non-vanishing expectation values.
	If it contributes none or one, $\Tr(W\sigma) \in \{0, \sfrac{1}{2}, 1\}$ in all cases.
	If it contributes three, the relations above fix $c = -ab$ and $c = de$ respectively, and
	\begin{equation}
		\begin{split}
			\Tr(W\sigma) & = \frac{a-b}{2} + \frac{1+ab}{2} \in \{1, a\},   \\
			\Tr(W\sigma) & = -\frac{d+e}{2} + \frac{1-de}{2} \in \{1, -d\}.
		\end{split}
	\end{equation}
	In all cases $|\Tr(W\sigma)| \leq 1$.
\end{proof}

\begin{proposition}[Closed form of the RoM]
	\label{prop:rom_closed}
	Let $\Lambda$ be a channel with $\sym(M_\Lambda) = M_\Lambda$, of reduced coordinates $(\zeta_\Lambda, s_\Lambda)$ as in \cref{app_eq:cone_pyramid_height}.
	Then
	\begin{equation}
		\label{app_eq:rom_closed}
		\mathcal R(J_\Lambda) - 1 = \max\qty(0, \, ||\zeta_\Lambda||_1 - s_\Lambda).
	\end{equation}
\end{proposition}
\begin{proof}
	As in \cref{app:aligned_optimum} we may assume $0 \leq \Im \zeta_\Lambda \leq \Re \zeta_\Lambda$.
	Indeed, composing $\Lambda$ with $S$ and conjugating it with $\sx$ act on the Choi state as conjugations by local Cliffords, which permute the pure stabilizer states and therefore leave $\mathcal R$ invariant, while on the coordinates they generate the eight symmetries of the diamond, which preserve $||\zeta_\Lambda||_1$ and $s_\Lambda$.

	\emph{Upper bound.}
	If $||\zeta_\Lambda||_1 \leq s_\Lambda$ the channel belongs to $\mathcal P$ and \cref{lemma:clifford_pyramid} provides a decomposition of weight one, so that $\mathcal R(J_\Lambda) = 1$.
	Otherwise $s_\Lambda > 0$ and the same construction applies with signed weights: the Clifford mixture \cref{app_eq:clifford_channel_pyramid} associated to $\zeta_\Lambda/s_\Lambda$, which lies in the unit disk because $M_\Lambda \in \mathcal C$, cf. \cref{app_eq:cone_pyramid_height}, has weight
	\begin{equation}
		\nu = s_\Lambda \sum_k |\beta_k| + (1 - s_\Lambda) = ||\zeta_\Lambda||_1 + 1 - s_\Lambda
	\end{equation}
	by \cref{app_eq:signed_diamond}, the channel $\mathcal A$ entering with the non-negative weight $1 - s_\Lambda$.
	\cref{lemma:signed_clifford} then gives $\mathcal R(J_\Lambda) \leq ||\zeta_\Lambda||_1 + 1 - s_\Lambda$.

	\emph{Lower bound.}
	By \cref{app_eq:choi_bloch,app_eq:reduced_coordinates} the only non-vanishing Pauli expectation values of $J_\Lambda$ are
	\begin{equation}
		\begin{split}
			 & \ev{\sx\otimes\sx} = -\ev{\sy\otimes\sy} = \Re \zeta_\Lambda, \\
			 & \ev{\sx\otimes\sy} = \ev{\sy\otimes\sx} = -\Im \zeta_\Lambda, \\
			 & \ev{\sz\otimes\sz} = 2 s_\Lambda - 1,
		\end{split}
	\end{equation}
	so that the witness \cref{app_eq:rom_witness} evaluates to
	\begin{equation}
		\Tr(W J_\Lambda) = \Re \zeta_\Lambda + \Im \zeta_\Lambda + 1 - s_\Lambda = ||\zeta_\Lambda||_1 + 1 - s_\Lambda,
	\end{equation}
	and \cref{lemma:rom_witness} gives $\mathcal R(J_\Lambda) \geq ||\zeta_\Lambda||_1 + 1 - s_\Lambda$.
	Together with $\mathcal R \geq 1$, this proves \cref{app_eq:rom_closed}.
\end{proof}

We can now evaluate \cref{app_eq:rom_closed} on the channels of \cref{sec:classical_phases}.
For $\phi = -\pi/8$ the coordinates \cref{app_eq:pauli_coordinates} read $\zeta_\Lambda = f_\bot e^{-i\pi/4}$ and $s_\Lambda = (1+f_z)/2$, so that $||\zeta_\Lambda||_1 = \sqrt2 |f_\bot|$ and
\begin{equation}
	\label{app_eq:rom_general}
	\mathcal R(J_\Lambda) = \max\qty(1, \, \sqrt2\, |f_\bot| + \frac{1 - f_z}{2}).
\end{equation}
Aligned dephasing has $f_\bot = 1-2p$ and $f_z = 1$, whence
\begin{equation}
	\mathcal R(J_\Lambda) = \max\qty(1, \, \sqrt 2 \, |1-2p|),
\end{equation}
while depolarizing noise has $f_\bot = f_z = q$ with $q = 1 - \frac{4}{3}p$, and splitting \cref{app_eq:rom_general} according to the sign of $q$,
\begin{equation}
	\mathcal R(J_\Lambda) = \max\qty(1, \, \frac{1 + (2\sqrt2-1)q}{2}, \, \frac{1 - (2\sqrt2+1)q}{2}).
\end{equation}
These are the two lines of \cref{eqs:robustness_closed}.

\cref{app_eq:rom_closed} also makes the geometric content of the RoM explicit: $\mathcal R(J_\Lambda) - 1$ is the excess of the channel over the facet $||\zeta||_1 = s$ of the Clifford pyramid \cref{app_eq:cone_pyramid_height}.
It vanishes exactly on $\mathcal P$, its level sets are the dilated pyramids $||\zeta||_1 - s = \rm{const}$, and its only non-analyticity is the facet crossing.
In particular, by \cref{lemma:clifford_pyramid} and case $(i)$ of \cref{thm:pauli_noise},
\begin{equation}
	\mathcal R(J_\Lambda) = 1 \iff (\zeta_\Lambda, s_\Lambda) \in \mathcal P \iff \bar c = 0,
\end{equation}
so that within this family the RoM is faithful on channels.
For aligned dephasing $s_\Lambda = 1$, the pyramid is cut at its base, and the picture reduces to the Clifford diamond of \cref{fig:clifford_diamond} with $\mathcal R(J_\Lambda) - 1 = ||z_\Lambda||_1 - 1$.

\subsection{From the robustness to the cost}
\label{app:rom_cost}

\begin{lemma}[Robustness of a Kraus operator]
	\label{lemma:rom_kraus}
	Let $K$ be a unitary Kraus operator of cost $c(M_K) = k$.
	Then $\mathcal R(J_K) \leq (\sqrt 2)^k$.
\end{lemma}
\begin{proof}
	The channel of a rotation $R_z(2\theta)$ has complex coordinate $z = e^{2i\theta}$ on the unit circle, which \cref{app_eq:signed_diamond} decomposes over $\mathcal G$ with weight $||z||_1 = |\cos(2\theta)| + |\sin(2\theta)| \leq \sqrt 2$.
	Clifford channels have weight one, so that composing the $k$ rotations and the Cliffords of the decomposition $M_K = \tilde C_1 R_z(2\theta_1) \tilde C_2 R_z(2\theta_2) \tilde C_3 R_z(2\theta_3) \tilde C_4$ of \cref{lemma:cost_criterion} through \cref{lemma:signed_clifford} gives a signed Clifford decomposition of $K$ of weight at most $(\sqrt2)^k$.
\end{proof}

\begin{proof}[Proof of \cref{thm:magic}]
	Let $M_\Lambda = \sum_m p_m M_m$, with $M_m \in \conv(\mathcal K_m)$, be feasible for \cref{eq:optimization_general}.
	Writing each $M_m$ as a convex combination of Bloch matrices of unitary Kraus operators of cost at most $m$, and using the linearity of $\Lambda \mapsto J_\Lambda$, the Choi state $J_\Lambda$ is a convex combination of Choi states $J_K$ with $\mathcal R(J_K) \leq (\sqrt 2)^m$ by \cref{lemma:rom_kraus}.
	The RoM is convex, as decompositions of the parts combine into a decomposition of the mixture, hence
	\begin{equation}
		\mathcal R(J_\Lambda) \leq \sum_m p_m \, (\sqrt2)^m.
	\end{equation}
	By \cref{lemma:cost_criterion} the costs never exceed three, and on $m \in \{0,1,2,3\}$ the convexity of $m \mapsto (\sqrt2)^m$ gives the chord bound
	\begin{equation}
		(\sqrt 2)^m \leq 1 + \frac{2\sqrt2 - 1}{3}\, m,
	\end{equation}
	with equality at the endpoints.
	Therefore $\mathcal R(J_\Lambda) - 1 \leq \frac{2\sqrt2-1}{3} \sum_m p_m m$, and minimizing the right-hand side over the feasible points of \cref{eq:optimization_general} yields $\sum_m p_m m = \bar c$ from which \cref{eq:magic_bound} follows directly.

	For the second statement, let $\Lambda$ be as in \cref{thm:pauli_noise} with $\phi = -\pi/8$, so that $A = B = |f_\bot|/\sqrt2$ and $A + B = ||\zeta_\Lambda||_1$.
	Case $(iii)$ of \cref{thm:pauli_noise} never occurs, since $A + (\sqrt2 - 1) B = \sqrt 2 A = |\zeta_\Lambda| \leq s_\Lambda$ by \cref{app_eq:cone_pyramid_height}.
	The first two cases of \cref{eq:cost_pauli_noise} then collapse into a single expression, which \cref{app_eq:rom_closed} rewrites as
	\begin{equation}
		\bar c = \frac{\max\qty(0, \, ||\zeta_\Lambda||_1 - s_\Lambda)}{\sqrt2 - 1} = \frac{\mathcal R(J_\Lambda) - 1}{\sqrt 2 - 1},
	\end{equation}
	that is, \cref{eq:relation_cost_rom}.
\end{proof}

\section{Proofs for CAMPO disentangling}
\label{sec:proofs_campo_disentangling}

This appendix proves \cref{thm:density_operator} and analyses the two distinct mechanisms that make the bond dimension of the inner MPO collapse to $\chi = 1$ in \cref{fig:campo_heatmap}.
We first recall the closed-system construction of Ref.~\cite{fux2025disentangling} in the only form we need (\cref{app:campo_noiseless}): in the frame of the disentangler, each layer is a single-qubit rotation on a fresh qubit, dressed by a stabilizer of the qubits not yet used.
Both effects of a Pauli insertion follow from this (\cref{app:campo_dressing}): it keeps every trajectory a product state in the \emph{same} frame, which proves the theorem, and it flips the unconsumed register, hence the eigenvalues $\epsilon = \pm 1$ of the dressing stabilizers on it, which is the origin of the classical correlations of $\Sigma$.
Aligned dephasing is the special case in which the inserted Pauli \emph{is} the gate axis: the sign flips are then absent and $\Sigma$ is an exact product state, $\chi = 1$ (\cref{app:campo_aligned}).
\cref{app:campo_infinite_temperature} deals with the opposite regime $N_T \gg N$, where $\rho$ has relaxed to $\mathds{1}/2^N$, and estimates the depth at which a partially relaxed $\rho$ is truncated to $\chi = 1$.

\subsection{The noiseless disentangler}
\label{app:campo_noiseless}

Write the noiseless skeleton of \cref{fig:sketch_circuit} and the magic states it produces as
\begin{equation}
	\label{app_eq:skeleton}
	W_s = U_\phi C_s \cdots U_\phi C_1, \qquad \ket{x(\theta)} = e^{i \theta \sy}\ket{0},
\end{equation}
with $U_\phi = e^{i\phi\sz_1}$.
Following Ref.~\cite{fux2025disentangling}, one constructs Clifford operators $\tilde C_s = C_s \tilde C_{s-1} V_s$, with $V_s$ Clifford and $\tilde C_0 = \mathds 1$, together with pairwise distinct \emph{slots} $j_1, \ldots, j_s$, such that in the frame of $\tilde C_s$ the $s$-th layer is a single-qubit rotation dressed by a stabilizer, followed by the residual Clifford $V_s^\dagger$,
\begin{equation}
	\label{app_eq:frame_localization}
	\tilde C_s^\dagger \, U_\phi C_s \, \tilde C_{s-1} = e^{i \phi \, Q_s} \, V_s^\dagger, \quad Q_s = \sy_{j_s} S_s,
\end{equation}
where $S_s$ belongs to the stabilizer group of the qubits that have not been consumed,
\begin{equation}
	\label{app_eq:running_code}
	\mathcal S_s = \big\langle \sz_j \, : \, j \notin \{j_1, \ldots, j_s\} \big\rangle,
\end{equation}
so that $[\sy_{j_s}, S_s] = 0$ and $Q_s^2 = \mathds 1$.
Since $U_\phi = e^{i \phi \sz_1}$ and $C_s \tilde C_{s-1} = \tilde C_s V_s^\dagger$, \cref{app_eq:frame_localization} is equivalent to the statement that $Q_s$ is the image of the gate axis in the frame,
\begin{equation}
	\label{app_eq:axis_image}
	\tilde C_s^\dagger \, \sz_1 \, \tilde C_s = Q_s.
\end{equation}
Three properties of the construction carry the whole generalization to open systems.
\begin{itemize}
	\item [$(i)$] The Cliffords $\tilde C_s$, the slots $j_s$ and the hypothesis under which they exist -- that the pulled-back gate axis $P_s \equiv (C_s \tilde C_{s-1})^\dagger \, \sz_1 \, (C_s \tilde C_{s-1})$ anticommutes with at least one generator of $\mathcal S_{s-1}$, i.e. that $P_s$ lies outside the normalizer $N(\mathcal S_{s-1})$ -- are fixed by the Clifford data $\{C_r\}_{r\leq s}$ and by the axis $\sz_1$ alone.
	      They depend neither on the angle $\phi$, nor on the state, nor on the \emph{signs} of the stabilizers.
	\item [$(ii)$] The residual Cliffords $V_s$ preserve the structure of the frame states: $V_s^\dagger$ maps a product state whose qubits outside the slots $\{j_1, \ldots, j_{s-1}\}$ are in a computational basis state $\ket b$ into another state of the same form, and acts as the identity when $\ket b = \ket 0^{\rm unused}$.
	\item [$(iii)$] For deep random Cliffords the hypothesis holds at every step as long as $N_T \lesssim N$~\cite{fux2025disentangling}.
\end{itemize}
The case $P_s \in \mathcal S_{s-1}$, in which $e^{i\phi P_s}$ acts on the state as a global phase, is harmless: no slot is consumed, $V_s = \mathds 1$, and the construction proceeds with $\tilde C_s = C_s \tilde C_{s-1}$. 
It only delays the breakdown of the recursion, and we disregard it in the following; the remaining case $P_s \in N(\mathcal S_{s-1}) \setminus \mathcal S_{s-1}$ is the genuine breakdown excluded by (iii).
The normal form \cref{app_eq:frame_localization} is fixed only up to $\mathcal S_s$: we adopt here the \emph{minimal} choice, in which $V_s$ strips from the pulled-back axis $P_s$ only what is not already an element of $\mathcal S_s$, whereas Ref.~\cite{fux2025disentangling} strips the Pauli string completely and obtains $S_s = \mathds 1$.
The two conventions differ by Clifford operations that act trivially on the noiseless state and give the same $\rho$; retaining $S_s$ is what makes the trajectory dependence explicit in \cref{app:campo_dressing}.
The closed-system result is \cref{app_eq:frame_localization} evaluated on $\ket{0}^{\otimes N}$: every $S_s$ acts there as the identity and every $V_s^\dagger$ acts as the identity by $(ii)$, so that each layer rotates its own fresh slot, and
\begin{equation}
	\label{app_eq:noiseless_product}
	W_{N_T}\ket{0}^{\otimes N} = \tilde C \qty(\bigotimes_{s=1}^{N_T} \ket{x(\phi)}_{j_s} \otimes \ket{0}^{\rm unused}),
\end{equation}
with $\tilde C = \tilde C_{N_T}$.

\subsection{Pauli dressing and proof of \cref{thm:density_operator}}
\label{app:campo_dressing}

All the Kraus operators of a Pauli channel are proportional to unitaries, so the probability $q_\nu$ of a trajectory $\nu$ -- a choice of a Kraus branch at each of the noise locations -- is a product of the noise probabilities and does not depend on the state.
The density operator is then the classical mixture
\begin{equation}
	\label{app_eq:trajectory_mixture}
	\begin{split}
		 & \rho = \sum_\nu q_\nu \ketbra{\psi_\nu},                                                         \\
		 & \ket{\psi_\nu} \propto E^\nu_{N_T} U_\phi C_{N_T} \cdots E^\nu_1 U_\phi C_1 \ket 0 ^{\otimes N},
	\end{split}
\end{equation}
where $E^\nu_s \in \mathcal P_N$ is the Pauli selected at layer $s$.

\begin{lemma}[Pauli dressing]
	\label{lemma:pauli_dressing}
	Let $N_T \lesssim N$.
	Then, for every trajectory $\nu$,
	\begin{equation}
		\label{app_eq:dressed_trajectory}
		\ket{\psi_\nu} \propto \tilde C \bigotimes_{j=1}^N \ket{y_{\nu, j}},
	\end{equation}
	with $\ket{y_{\nu, j}}$ single-qubit pure states and $\tilde C$ the noiseless disentangler of \cref{app_eq:noiseless_product}, the same for every trajectory.
	The consumed slots carry $\ket{y} \in \{\ket{x(\pm\phi + m \pi/2)}\}$, the unused ones $\ket{y} \in \{\ket 0, \ket 1\}$.
\end{lemma}
\begin{proof}
	We show by induction that, in the frame of the current $\tilde C_s$, the state is a product whose unconsumed qubits are in a computational basis state $\ket b$.
	This holds at $s = 0$.
	A layer acts through \cref{app_eq:frame_localization}: first the residual $V_s^\dagger$, which by $(ii)$ preserves the product structure and leaves the unconsumed register in a computational basis state, again denoted $\ket b = \ket{b_{j_s}} \otimes \ket*{\hat b}$, with $\ket*{\hat b}$ its restriction to the qubits that remain unconsumed after the layer; then $e^{i \phi Q_s}$, which on such a state reduces, using $S_s \ket*{\hat b} = \epsilon \ket*{\hat b}$ with $\epsilon = \pm 1$, to the single-qubit rotation
	\begin{equation}
		\label{app_eq:flipped_layer}
		e^{i\phi Q_s} \qty(\ket{b_{j_s}} \otimes \ket*{\hat b})
		= \qty(e^{i \epsilon \phi \sy_{j_s}}\ket{b_{j_s}}) \otimes \ket*{\hat b}
	\end{equation}
	of the fresh slot, and leaves $\ket*{\hat b}$ untouched on the qubits that remain unconsumed.
	An insertion acts through $\tilde C_s^\dagger E^\nu_s \tilde C_s \in \mathcal P_N$, a product of single-qubit Paulis: on the slots it permutes the magic states within the finite orbit
	\begin{equation}
		\label{app_eq:pauli_orbit}
		\begin{split}
			\sz\ket{x(\theta)} & = \ket{x(-\theta)},              \\
			\sy\ket{x(\theta)} & \propto \ket{x(\theta + \pi/2)}, \\
			\sx\ket{x(\theta)} & \propto \ket{x(\pi/2 - \theta)},
		\end{split}
	\end{equation}
	while on the unconsumed register it maps $\ket b$ to another basis state $\ket{b'}$.
	Both operations preserve the product structure and the induction hypothesis, and by $(i)$ the frames $\tilde C_s$ are the ones of the noiseless skeleton.
\end{proof}

\begin{proof}[Proof of \cref{thm:density_operator}]
	Inserting \cref{app_eq:dressed_trajectory} into \cref{app_eq:trajectory_mixture},
	\begin{equation}
		\rho = \tilde C \, \Sigma \, \tilde C^\dagger, \qquad \Sigma = \sum_\nu q_\nu \bigotimes_{j=1}^N \ketbra{y_{\nu, j}},
	\end{equation}
	which is the claimed form with $\sigma^\nu_j = \ketbra{y_{\nu,j}}$.
	The only hypothesis used is $(iii)$, which constrains the Clifford data alone: the threshold is therefore the closed-system one, for Pauli noise of any strength at any position.
	Each term of $\Sigma$ has zero entanglement, and the noise only distributes the weights $q_\nu$ over the products, i.e. it only generates classical correlations.
\end{proof}

The bond dimension that $\Sigma$ costs is the correlation between the labels $y_{\nu, j}$ of different sites, and \cref{app_eq:flipped_layer} shows where it comes from.
An insertion at layer $r$ flips the unconsumed register, $\ket b \to \ket{b'}$, and the flip is transported to every subsequent layer by the residual Cliffords; it changes the sign $\epsilon$ of each later layer whose dressing $S_s$ is supported on an odd number of the flipped qubits, which whenever the residual Cliffords leave the register alone is simply the condition that $S_s$ anticommutes with the insertion.
Either way, a single noise event correlates the angles of all the downstream slots.

\subsection{Aligned dephasing: $\Sigma$ is a product state}
\label{app:campo_aligned}

\begin{proposition}[Product structure for aligned dephasing]
	\label{prop:campo_aligned}
	For the aligned dephasing channel \cref{eq:aligned_dephasing} acting on the gate qubit and $N_T \lesssim N$,
	\begin{equation}
		\label{app_eq:campo_aligned_product}
		\begin{gathered}
			\Sigma = \bigotimes_{s=1}^{N_T} \sigma_{j_s} \otimes \ketbra{0}^{\rm unused},                             \\
			\sigma_{j_s} = (1-p)\ketbra{x(\phi)} + p \, \sy \ketbra{x(\phi)} \sy.
		\end{gathered}
	\end{equation}
	The slot states have spectrum $\{1-p, \, p\}$ and Bloch vector of length $|f| = |1-2p|$, and the inner MPO has $\chi = 1$.
\end{proposition}
\begin{proof}
	The Kraus operators of the aligned dephasing are $\mathds 1$ and $\sz_1$, the gate axis itself, so by \cref{app_eq:axis_image} the insertion at layer $s$ acts in the frame as $Q_s = \sy_{j_s} S_s$.
	The unconsumed register is never changed: the layers and the residual Cliffords do not flip it, by \cref{app_eq:flipped_layer} and property $(ii)$, and the insertions act on it through $S_s \in \mathcal S_s$, a product of $\sz_j$ that leaves $\ket{0}^{\rm unused}$ invariant.
	Hence $\epsilon = 1$ at every step, $S_s$ acts as the identity, and the insertion reduces to the single-qubit Pauli $\sy_{j_s}$ on its own slot.
	Different layers thus act on different slots, through independent Bernoulli variables of parameter $p$, and averaging them gives \cref{app_eq:campo_aligned_product}.
	Finally $\ev{\sy}{x(\phi)} = \ev{\sy}{0} = 0$, so the two states in $\sigma_{j_s}$ are orthogonal and the Bloch vector of $\sigma_{j_s}$ is $(1-2p)$ times the one of $\ket{x(\phi)}$.
\end{proof}

Aligned dephasing is thus the extreme case of \cref{thm:density_operator}: the noise is invisible to the disentangler \emph{and} it creates no classical correlations, so that CAMPO reproduces $\rho$ exactly with $\chi = 1$ up to the closed-system threshold, at any $p$.
This is what \cref{fig:campo_heatmap} shows, the boundary of the disentanglable region at $N_T \sim N$ being independent of $p$.
Any other Pauli insertion -- a different axis, or the same axis on a different qubit -- has a frame image that does flip the register, and $\Sigma$ pays bond dimension for the correlated angle flips.

\subsection{Relaxation to $\rho \propto \mathds 1$ and the truncation threshold}
\label{app:campo_infinite_temperature}

The black region of \cref{fig:campo_heatmap} at $N_T \gtrsim N$ has a different origin, which we now quantify.
Expand the density operator on the Pauli basis, $\rho = 2^{-N}\sum_{P \in \mathcal P_N} c_P \, P$, and group the strings in quadruplets $\alpha_1 \otimes Q$ sharing the same tail $Q$, with $\alpha \in \{\mathds 1, \sx, \sy, \sz\}$.
A Clifford permutes the strings, and their coefficients, up to a sign.
The noisy gate acts instead through the PTM \cref{eq:aligned_dephasing_ptm}: it leaves $c_{\mathds 1 \otimes Q}$ and $c_{\sz \otimes Q}$ invariant, while
\begin{equation}
	\label{app_eq:damping}
	\begin{pmatrix} c_{\sx \otimes Q} \\ c_{\sy \otimes Q} \end{pmatrix} \longrightarrow f \, R(2\phi) \begin{pmatrix} c_{\sx \otimes Q} \\ c_{\sy \otimes Q} \end{pmatrix}, \qquad f = 1 - 2p.
\end{equation}
A string is therefore damped by $|f|$ exactly when it carries $\sx$ or $\sy$ on the first qubit.
A global random Clifford maps a non-identity string onto a uniformly random non-identity string, so that this happens with probability
\begin{equation}
	\label{app_eq:half_probability}
	\mathbb P\qty[\alpha \in \{\sx, \sy\}] = \frac{2 \cdot 4^{N-1}}{4^N - 1} \simeq \frac 12,
\end{equation}
independently at each layer.
After $N_T$ layers $|c_P| \simeq |f|^{k}$ with $k \sim \mathrm{Bin}(N_T, 1/2)$, whence
\begin{equation}
	\label{app_eq:mean_damping}
	\mathbb E \qty[\,|c_P|\,] \simeq \qty(\frac{1 + |f|}{2})^{N_T} \xrightarrow[N_T \to \infty]{} 0, \qquad 0 < p < 1.
\end{equation}
Only $c_{\mathds 1} = 1$ survives, and $\rho \to \mathds 1/2^N$: an MPO of bond dimension one.
This is a relaxation and not a disentangling, as it is the state, and not the Clifford operator of the Ansatz, that has become simple.

At finite truncation error $\varepsilon$ the collapse to $\chi = 1$ occurs at finite depth.
Across any bipartition the singular values of the MPO are those of the matrix of coefficients $c_{P_A \otimes P_B}$: the identity carries the leading one, and all the others are of the order of the surviving coefficients.
A string is thus discarded as soon as $|f|^k \leq \varepsilon$, that is, as soon as it has been damped
\begin{equation}
	\label{app_eq:kstar}
	k \geq k_*(p) = \frac{\log \varepsilon}{\log|1-2p|}
\end{equation}
times, which at fixed depth $N_T$ happens with probability
\begin{equation}
	\label{app_eq:truncation_probability}
	\mathbb P(N_T, p) = \sum_{k \geq k_*(p)} \binom{N_T}{k} 2^{-N_T}.
\end{equation}
The typical string is truncated when $k_*(p) \leq N_T/2$, which delimits the region
\begin{equation}
	\label{app_eq:campo_threshold}
	p_-(N_T) \leq p \leq p_+(N_T), \qquad p_\pm(N_T) = \frac{1 \pm \varepsilon^{2/N_T}}{2},
\end{equation}
reported in \cref{fig:campo_heatmap}.
The multiplicity of the $\sim 2^N$ damped strings only renormalizes $\varepsilon$ by a depth-independent factor, and shifts \cref{app_eq:campo_threshold} at subleading order.
The region opens up with the depth, since $\varepsilon^{2/N_T} \to 1$, and closes onto $p = 1/2$ as $\varepsilon \to 0$: it is an artifact of the truncation, and would disappear altogether in an exact simulation, consistently with \cref{fig:campo_scaling}.

\begin{figure}[t!]
	\centering
	\includegraphics{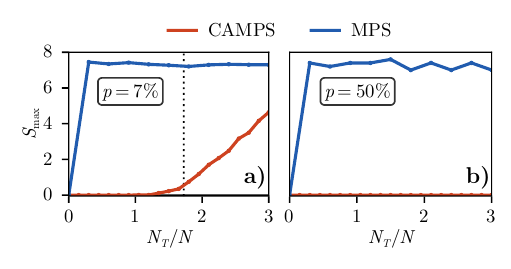}
	\caption{Comparison between a simulation with MPS and one with CAMPS.
		(a) -- Comparison of the maximal entanglement entropy $S_{\max}$ at $p = 7\%$ computed on the inner state of a CAMPS, as well as a plain MPS.
		After the first layer the simple MPS is already at its maximal bond dimension: the Clifford operator scrambles information maximally.
		(b) -- Same comparison as in panel (a), in the classical region now.
		MPS still necessitates large bond dimension.
		CAMPS on the other hand remains at zero entropy throughout the simulation.
	}
	\label{fig:mps_comparison}
\end{figure}

\section{MPS simulations of the noise channels}
\label{sec:mps}

Here we provide numerical evidence that the disentangled regions visible with the optimal unraveling that we discussed throughout the paper do not have a similar origin as the truncation region of \cref{app:campo_infinite_temperature}.
We simulate the same circuits as in \cref{sec:classical_phases}, with a plain MPS Ansatz.
That is, we do not try to disentangle the state with Clifford operators, but we simply compress the MPS by truncating the singular values of the Schmidt decomposition across each bond.
The results are shown in \cref{fig:mps_comparison}.
As expected, the random global Clifford gates generate volume-law entanglement already after the first layer.
This results in the MPS bond dimension $\chi_{\rm MPS}$ and maximal entanglement entropy $S_{\rm max}$ reaching their maximum values, which are limited only by the system size $N$.
On the other hand, the CAMPS Ansatz is able to disentangle the state in those same regions, and the bond dimension $\chi$ remains small.
This confirms that the disentangled regions are not due to the state necessitating only a small bond dimension, but rather due to the ability of the CAMPS Ansatz to effectively disentangle the state.

\bibliography{bibliography}

\end{document}